\documentclass[preprint,3p,12pt]{elsarticle}
\usepackage{lineno}
\modulolinenumbers[5]

\journal{Journal of the Franklin Institute}

\usepackage{lmodern}
\usepackage[utf8]{inputenc}
\usepackage[american]{babel}
\addto\captionsamerican{} 

\usepackage{amsthm}
\usepackage{amsmath}
\usepackage{amssymb}
\usepackage{amsfonts}
\usepackage{bm}
\usepackage{hhline}
\usepackage{csquotes}
\usepackage{color}
\usepackage{import}
\usepackage{algorithm}
\usepackage{algpseudocode}
\usepackage{dsfont}
\usepackage{multirow}
\usepackage{tensor}
\usepackage{subfig}

\usepackage{graphicx}
\usepackage{microtype}
\usepackage{dsfont}

\makeatletter
\usepackage[colorlinks=true,linkcolor=red]{hyperref}
\makeatother

\biboptions{numbers,sort&compress}

\newtheorem{theorem}{Theorem}
\newtheorem{proposition}{Proposition}

\allowdisplaybreaks

\newcommand{\textred}[1]{#1}

\newcommand{\ten}[1]{\ensuremath{{\cdot}10^{#1}}}

\newcommand{\frI}{\ensuremath{\mathcal{I}}}
\newcommand{\frL}{\ensuremath{\mathcal{L}}}
\newcommand{\frB}{\ensuremath{\mathcal{B}}}
\newcommand{\frC}[1]{\ensuremath{{\mathcal{C}_{#1}}}}
\newcommand{\frA}[1]{\ensuremath{{\mathcal{A}_{#1}}}}
\newcommand{\RIB}{\ensuremath{\bm{R}^\mathcal{I}_\mathcal{B}}}
\newcommand{\RIL}{\ensuremath{\bm{R}^\mathcal{I}_\mathcal{L}}}
\newcommand{\RLA}[1]{\ensuremath{\bm{R}^\mathcal{L}_{\mathcal{A}_{#1}}}}
\newcommand{\RAB}[1]{\ensuremath{\bm{R}^{\mathcal{A}_{#1}}_\mathcal{B}}}
\newcommand{\RLB}{\ensuremath{\bm{R}^\mathcal{L}_\mathcal{B}}}
\newcommand{\RBC}[1]{\ensuremath{\bm{R}^\mathcal{B}_{\mathcal{C}_{#1}}}}
\newcommand{\RBA}[1]{\ensuremath{\bm{R}^\mathcal{B}_{\mathcal{A}_{#1}}}}
\newcommand{\RAA}[2]{\ensuremath{\bm{R}^{\mathcal{A}_{#1}}_{\mathcal{A}_{#2}}}}
\newcommand{\RAC}[2]{\ensuremath{\bm{R}^{\mathcal{A}_{#1}}_{\mathcal{C}_{#2}}}}

\newcommand{\dBC}[1]{\ensuremath{\bm{d}^\mathcal{B}_{\mathcal{C}_{#1}}}}
\newcommand{\dLA}[1]{\ensuremath{\bm{d}^\mathcal{L}_{\mathcal{A}_{#1}}}}
\newcommand{\dBA}[1]{\ensuremath{\bm{d}^\mathcal{B}_{\mathcal{A}_{#1}}}}
\newcommand{\dAB}[1]{\ensuremath{\bm{d}^{\mathcal{A}_{#1}}_{\mathcal{B}}}}
\newcommand{\dAA}[2]{\ensuremath{\bm{d}^{\mathcal{A}_{#1}}_{\mathcal{A}_{#2}}}}
\newcommand{\dAC}[2]{\ensuremath{\bm{d}^{\mathcal{A}_{#1}}_{\mathcal{C}_{#2}}}}

\newcommand{\pLL}{\ensuremath{\bm{p}^\mathcal{L}_\mathcal{L}}}
\newcommand{\pii}[1]{\ensuremath{\bm{p}^{\mathcal{C}_{#1}}_{\mathcal{C}_{#1}}}}
\newcommand{\pIL}{\ensuremath{\bm{p}^\mathcal{I}_\mathcal{L}}}
\newcommand{\pIi}[1]{\ensuremath{\bm{p}^\mathcal{I}_{\mathcal{C}_{#1}}}}
\newcommand{\dpIL}{\ensuremath{\dot{\bm{p}}^\mathcal{I}_\mathcal{L}}}
\newcommand{\dpIi}[1]{\ensuremath{\dot{\bm{p}}^\mathcal{I}_{\mathcal{C}_{#1}}}}

\newcommand{\oIi}[1]{\ensuremath{\bm{o}^\mathcal{I}_{\mathcal{C}_#1}}}
\newcommand{\doIL}{\ensuremath{\dot{\bm{o}}^\mathcal{I}_\mathcal{L}}}
\newcommand{\doIi}[1]{\ensuremath{\dot{\bm{o}}^\mathcal{I}_{\mathcal{C}_#1}}}

\newcommand{\dxi}{\ensuremath{\dot{\bm{\xi}}}}
\newcommand{\deta}{\ensuremath{\dot{\bm{\eta}}}}
\newcommand{\dgamma}{\ensuremath{\dot{\bm{\gamma}}}}
\newcommand{\gone}{\ensuremath{{\gamma_1}}}

\newcommand{\gtwo}{\ensuremath{{\gamma_2}}}

\newcommand{\aR}{\ensuremath{{\alpha_\text{R}}}}
\newcommand{\daR}{\ensuremath{\dot{\alpha}_\text{R}}}
\newcommand{\aL}{\ensuremath{{\alpha_\text{L}}}}
\newcommand{\daL}{\ensuremath{\dot{\alpha}_\text{L}}}

\newcommand{\wLIL}{\ensuremath{\bm{\omega}^\mathcal{L}_{\mathcal{I}\mathcal{L}}}}

\newcommand{\wALA}[1]{\ensuremath{\bm{\omega}^{\mathcal{A}_{#1}}_{\mathcal{L} \mathcal{A}_{#1}}}}
\newcommand{\wBAB}[1]{\ensuremath{\bm{\omega}^\mathcal{B}_{\mathcal{A}_{#1}\mathcal{B}}}}
\newcommand{\wBLB}{\ensuremath{\bm{\omega}^\mathcal{B}_{\mathcal{L}\mathcal{B}}}}
\newcommand{\wBIB}{\ensuremath{\bm{\omega}^\mathcal{B}_{\mathcal{I}\mathcal{B}}}}
\newcommand{\wABA}[1]{\ensuremath{\bm{\omega}^{\mathcal{A}_{#1}}_{\mathcal{B} \mathcal{A}_{#1}}}}
\newcommand{\wiAi}[1]{\ensuremath{\bm{\omega}^{\mathcal{C}_{#1}}_{\mathcal{A}_{#1} \mathcal{C}_{#1}}}}
\newcommand{\wiBi}[1]{\ensuremath{\bm{\omega}^{\mathcal{C}_{#1}}_{\mathcal{B} \mathcal{C}_{#1}}}}

\newcommand{\mL}{\ensuremath{m_\mathcal{L}}}
\newcommand{\mi}[1]{\ensuremath{m_{#1}}}
\newcommand{\IL}{\ensuremath{\bm{I}_\mathcal{L}}}
\newcommand{\Ii}[1]{\ensuremath{\bm{I}_{#1}}}
\newcommand{\Ji}[1]{\ensuremath{\bm{J}_{#1}}}
\newcommand{\Ei}[1]{\ensuremath{\bm{E}_{#1}}}
\newcommand{\Di}[1]{\ensuremath{\bm{D}_{#1}}}

\newcommand{\half}{\ensuremath{\frac{1}{2}}}

\newcommand{\dm}{\ensuremath{\bm{d}_m}}

\newcommand{\Kin}{\ensuremath{\mathcal{K}}}
\newcommand{\KinL}{\ensuremath{\mathcal{K}_\mathcal{L}}}

\newcommand{\Pot}{\ensuremath{\mathcal{U}}}
\newcommand{\PotL}{\ensuremath{\mathcal{U}_\mathcal{L}}}
\newcommand{\Poti}[1]{\ensuremath{\mathcal{U}_{#1}}}

\newcommand{\Weta}{\ensuremath{\bm{W}_{\bm{\eta}}}}
\newcommand{\ay}{\ensuremath{\bm{a}_y}}
\newcommand{\az}{\ensuremath{\bm{a}_z}}

\newcommand{\eye}[1]{\ensuremath{\mathds{I}_{#1\times#1}}}

\newcommand{\zeros}[2]{\ensuremath{\bm{0}_{#1\times#2}}}
\newcommand{\ones}[2]{\ensuremath{\bm{1}_{#1\times#2}}}

\newcommand{\ktaub}{\ensuremath{\frac{k_\tau}{b}}}
\newcommand{\lambdaR}{\ensuremath{\lambda}_\text{R}}
\newcommand{\lambdaL}{\ensuremath{\lambda}_\text{L}}

\newcommand{\LambdaR}{\ensuremath{\bm{\Lambda}_\text{R}}}
\newcommand{\LambdaL}{\ensuremath{\bm{\Lambda}_\text{L}}}
\newcommand{\GammaR}{\ensuremath{\bm{\Gamma}_\text{R}}}
\newcommand{\GammaL}{\ensuremath{\bm{\Gamma}_\text{L}}}

\newcommand{\Jac}{\ensuremath{\mathcal{\bm{J}}}}
\newcommand{\GJac}{\ensuremath{\mathcal{\bm{W}}}}

\newcommand{\gforce}{\ensuremath{\bm{\vartheta}}}

\newcommand{\fR}{\ensuremath{\bm{f}_\text{R}}}
\newcommand{\fL}{\ensuremath{\bm{f}_\text{L}}}
\newcommand{\tauaR}{\ensuremath{\bm{\tau}_{\aR}}}
\newcommand{\tauaL}{\ensuremath{\bm{\tau}_{\aL}}}

\newcommand{\infR}{\ensuremath{f_\text{R}}}
\newcommand{\infL}{\ensuremath{f_\text{L}}}
\newcommand{\intauaR}{\ensuremath{\tau_{\aR}}}
\newcommand{\intauaL}{\ensuremath{\tau_{\aL}}}

\newcommand{\setreal}{\ensuremath{\mathbb{R}}}
\newcommand{\setrealmat}[2]{\ensuremath{\mathbb{R}^{#1\times#2}}}
\newcommand{\setrealvec}[1]{\ensuremath{\mathbb{R}^{#1}}}

\newcommand{\zdomain}{\ensuremath{\varsigma}}
\newcommand{\trace}[1]{\ensuremath{\text{tr}\left\{{#1}\right\}}}

\newcommand{\Achi}{\ensuremath{\bm{A}_{\bm{\chi}}}}
\newcommand{\Bchi}{\ensuremath{\bm{B}_{\bm{\chi}}}}
\newcommand{\Fchi}{\ensuremath{\bm{F}_{\bm{\chi}}}}

\newcommand{\intval}[1]{\ensuremath{{\tensor*[_{\llcorner}]{#1}{^{\urcorner}}}}}
\newcommand{\lbound}[1]{\ensuremath{\underline{#1}}}
\newcommand{\ubound}[1]{\ensuremath{\overline{#1}}}
\newcommand{\midpoint}[1]{\ensuremath{\text{mid}({#1})}}

\newcommand{\diam}[1]{\ensuremath{\text{diam}({#1})}}

\newcommand{\setintval}{\ensuremath{\mathbb{I} \mathbb{R}}}

\newcommand{\ubox}{\ensuremath{\mathbb{B}}}
\newcommand{\zonotope}{\ensuremath{\mathbb{Z}}}
\newcommand{\zonotopefamily}{\ensuremath{{\zonotope_\intval{}}}}
\newcommand{\strip}{\ensuremath{\mathbb{S}}}

\newcommand{\iextension}[2]{\ensuremath{\square_\text{#1} \left\{ #2 \right\} }}
\newcommand{\zinclusion}[1]{\ensuremath{\diamond \{ {#1} \}}}

\newcommand{\xiB}{\ensuremath{{\bm{\xi}_\frB}}}
\newcommand{\phiB}{\ensuremath{{\phi_\frB}}}
\newcommand{\thetaB}{\ensuremath{{\theta_\frB}}}
\newcommand{\psiB}{\ensuremath{{\psi_\frB}}}
\newcommand{\etaB}{\ensuremath{{\bm{\eta}_\frB}}}

\newcommand{\Anu}{\ensuremath{\bm{A}_{\bm{\nu}}}}
\newcommand{\Bnu}{\ensuremath{\bm{B}_{\bm{\nu}}}}

\newcommand{\Hnu}{\ensuremath{\bm{H}_{\bm{\nu}}}}

\newcommand{\expect}[1]{\ensuremath{\mathcal{E}\left[#1\right]}}

\newcommand{\setI}{\ensuremath{\mathbb{I}}}

\newcommand{\deps}{\ensuremath{\varepsilon}}
\newcommand{\real}[1]{\ensuremath{\text{Re}(#1)}}
\newcommand{\imag}[1]{\ensuremath{\text{Im}(#1)}}

\newcommand{\kwix}{\ensuremath{\bm{\kappa}^\frI_{\text{W},x}}}
\newcommand{\kwiy}{\ensuremath{\bm{\kappa}^\frI_{\text{W},y}}}

\begin{document}

\begin{frontmatter}


\title{Suspended Load Path Tracking Control Using a Tilt-rotor UAV Based on Zonotopic State Estimation\tnoteref{mytitlenote}}
\tnotetext[mytitlenote]{{\copyright} 2018. This manuscript version is made available under the CC-BY-NC-ND 4.0 license \url{https://creativecommons.org/licenses/by-nc-nd/4.0/}. This work was supported by the Brazilian agencies CNPq, CAPES, and FAPEMIG. The authors would like to acknowledge the InSAC - National Institute of Science and Technology for Cooperative Autonomous Systems Applied to Security and Environment.}


\author[myfirstaddress]{Brenner S. Rego\corref{mycorrespondingauthor}}
\cortext[mycorrespondingauthor]{Corresponding author}
\ead{brennersr7@ufmg.br}

\author[myfirstaddress,mysecondaddress]{Guilherme V. Raffo}
\ead{raffo@ufmg.br}

\date{July 18, 2018}

\address[myfirstaddress]{Graduate Program in Electrical Engineering, Federal University of Minas Gerais, Belo Horizonte, MG 31270-901, Brazil}
\address[mysecondaddress]{Department of Electronics Engineering, Federal University of Minas Gerais, Belo Horizonte, MG 31270-901, Brazil}

\begin{abstract}
This work addresses the problem of path tracking control of a suspended load using a tilt-rotor UAV. The main challenge in controlling this kind of system arises from the dynamic behavior imposed by the load, which is usually coupled to the UAV by means of a rope, adding unactuated degrees of freedom to the whole system. Furthermore, to perform the load transportation it is often needed the knowledge of the load position to accomplish the task. Since available sensors are commonly embedded in the mobile platform, information on the load position may not be directly available. To solve this problem in this work, initially, the kinematics of the multi-body mechanical system are formulated from the load's perspective, from which a detailed dynamic model is derived using the Euler-Lagrange approach, yielding a highly coupled, nonlinear state-space representation of the system, affine in the inputs, with the load's position and orientation directly represented by state variables. A zonotopic state estimator is proposed to solve the problem of estimating the load position and orientation, which is formulated based on sensors located at the aircraft, with different sampling times, and unknown-but-bounded measurement noise. To solve the path tracking problem, a discrete-time mixed $\mathcal{H}_2/\mathcal{H}_\infty$ controller with pole-placement constraints is designed with guaranteed time-response properties and robust to unmodeled dynamics, parametric uncertainties, and external disturbances. Results from numerical experiments, performed in a platform based on the Gazebo simulator and on a Computer Aided Design (CAD) model of the system, are presented to corroborate the performance of the zonotopic state estimator along with the designed controller.
\end{abstract}



\end{frontmatter}


\section{Introduction} \label{sec:introduction}

The problem of slung load transportation arises in a variety of essential tasks, such as transportation of containers in harbors \cite{Ngo2012}, aerial delivery of supplies in search-and-rescue missions \cite{Bernard2011}, and landmine detection \cite{BisgaardThesis}. The suspended load is usually connected to the mobile platform by means of a rope, considerably changing its dynamic behavior and adding unactuated degrees of freedom to the whole system. Moreover, the rope is a non-rigid body and is not always taut, which increases the task challenge. Several studies can be found in the literature, concerning different modeling approaches and control strategies for load transportation using overhead cranes \cite{Wu2015}, robotic manipulators \cite{Chen2007}, and aerial vehicles \cite{Bisgaard2009b,Fusato2001}.

An important issue in slung load transportation is the recurrent necessity of knowing the load position to accomplish the task, \textred{mainly when precise positioning of the load is required}. Since available sensors are often embedded in the mobile platform, information on the load position may not be directly obtained. The problem of estimating the load position then arises, being commonly addressed through visual systems and Bayesian state estimators. The Kalman filter is employed in \cite{JainThesis} for state estimation of a quadrotor unmanned aerial vehicle (UAV) with suspended load, in which measurements are provided by external cameras and sensors embedded at the aircraft. \textred{Considering a helicopter with suspended load platform, \cite{Bisgaard2007a} designs a data fusion algorithm based on the unscented Kalman filter (UKF) to estimate the load's position and velocity with measurements from an \textred{inertia measurement unit} (IMU) and a vision system, both located at the helicopter}. In \cite{Bisgaard2007b}, algorithms based on the UKF are proposed for estimation of the full state vector of a helicopter with suspended load, with measurements provided by a Global Positioning System (GPS), a magnetometer, a camera, an IMU on the helicopter and another one on the load. Kalman filtering algorithms require knowledge on statistical properties of existing process and measurement disturbances, which may not be easily obtained. In view of the exposed, the present work pursues set-membership estimation approaches, which require knowledge only on bounds of existing disturbances. These techniques are based on the construction of sets that include, with guarantee, the system states consistent with available measurements \cite{Alamo2005a,Le2013}. This work extends the zonotopic state estimation strategy proposed in \cite{Alamo2005a} to receive measurements provided by sensors with different sampling times.

The versatility and autonomous operation of UAVs are useful advantages in aerial load transportation. The main control design objectives in the literature include path tracking of the UAV with load swing attenuation~\cite{Bisgaard2009b,Palunko2012b,Almeida2015b,Raffo2016,Santos2016b,Oktay2013,Liang2018,Sanchez2017}, obstacle avoidance~\cite{CourHarbo2009,Tang2015}, transportation by multiple aircrafts~\cite{Bernard2009,Lee2013}, and trajectory tracking of the suspended load~\cite{Palunko2013,Sreenath2013b,Pereira2016b}. \textred{This paper focuses on the latter, which is the appropriate goal in tasks requiring precise maneuvering of the load. In contrast to the swing attenuation problem, the knowledge on the load position is usually required for such purpose}. A model-free, open-loop approach based on trajectory generation by machine learning is proposed in \cite{Palunko2013} for path tracking of a suspended load using a quadrotor UAV. However, the lack of a feedback structure prevents compensation of external disturbances affecting the load. A nonlinear cascade control strategy is designed in \cite{Sreenath2013b}, based on model decoupling, for trajectory tracking of a suspended load using a quadrotor UAV. Nevertheless, compensation of unmodelled dynamics and external disturbances is not addressed in the proposed strategy, and convergence issues are well known for cascade control systems. Assuming the aircraft as a system actuated by total thrust and orientation, \cite{Pereira2016b} \textred{proposes} another nonlinear solution to the problem of suspended load path tracking using a quadrotor UAV. Nevertheless, such assumption is valid for a very limited repertory of mechanical systems, which do not comprise the \textred{convertible} UAV configuration addressed by this work.

Most of the unmanned aerial vehicles used in load transportation tasks are in helicopter and quadrotor configurations. These rotary-wing UAVs have vertical take-off and landing (VTOL) and hovering capabilities, and achieve high maneuverability in low velocities. However, due to their limited flight envelope, such UAVs are not appropriate for missions that require long distance traveling, such as deployment of supplies to risky zones. To overcome such constraint, \textred{researches} are looking into the design of small-scale \textred{convertible} aircrafts, being the tilt-rotor configuration among the most popular ones \cite{Amiri2011,Park2013,Cardoso2016}. Provided with both fixed and rotary wings, tilt-rotor UAVs achieve an enlarged flight envelope by switching between helicopter and airplane flight-modes \textred{through tilting of the thrusters}. However, such advantages come with several design and control challenges, since these aircrafts are complex, underactuated mechanical systems with highly coupled dynamics. Additionally, when these UAVs are connected to a payload through a rope, the dynamic behavior of the system varies due to the load’s swing, which can destabilize the whole system if it is not well attenuated. A model predictive control (MPC) strategy is designed in \cite{Santos2016b} for path tracking of a tilt-rotor UAV with suspended load, in which the aircraft tracked a desired trajectory, while the load remained stable. A cascade strategy composed of three levels of feedback linearization controllers is proposed in \cite{Almeida2015b}, for trajectory tracking of a tilt-rotor UAV with load swing attenuation. The problem of path tracking of a suspended load using a tilt-rotor UAV is solved in \cite{Santos2017}, in which a model predictive controller is designed, taking into account time-varying load's mass and rope's length, and estimating the load's position and orientation by means of an unscented Kalman filter. However, the state estimation is not guaranteed, and nothing can be said about the transient response of the closed-loop system. The present work addresses the problem of trajectory tracking of a suspended load using a tilt-rotor UAV as mobile platform, with guaranteed time-response properties, compensation of unmodelled dynamics and external disturbances, and state estimation of the load's position and orientation, based on the set-membership approach to perform the task.

This paper is an extended, consolidated version of the previous work presented in \cite{Rego2016c}. To solve the aforementioned challenges, this work develops the whole-body dynamic equations of a tilt-rotor UAV with suspended load, from the perspective of the load. The position and orientation of the latter are chosen as degrees of freedom of the system, yielding a nonlinear state-space representation with these variables among the system states. As shown in \cite{Rego2016c,Santos2017}, this choice allows state-feedback control strategies to steer the trajectory of the load with respect to an inertial reference frame. In contrast to previous works, a reduced number of assumptions is made with respect to the physical system. This work designs a discrete-time state-feedback mixed $\mathcal{H}_2/\mathcal{H}_\infty$ control strategy with an enlarged domain of attraction for path tracking of the suspended load with disturbance rejection and guaranteed time-response properties, taking into account the desired accelerations of the load in the control design through an uncertain linear parameter-varying framework. In addition, this work proposes a zonotopic state estimation strategy to estimate the load's position and orientation when available measurements are provided by sensors with different sampling times. To demonstrate and compare the performance of the proposed state estimator, a Kalman filter is also designed. The performance of the proposed strategies are demonstrated through numerical experiments, performed in a platform based on the Gazebo simulator and on a Computer Aided Design (CAD) model of the system. The contributions of this work can be summarized as: (i) a detailed modeling from the load's point of view that comprises the dynamic coupling of the load and the tilt-rotor UAV, with few assumptions on the system, leading to an input-affine state-space representation with the load's position and orientation as state variables; (ii) a set-membership state estimation strategy based on zonotopes to provide the load's position and orientation, formulated for measurements with different sampling times and unknown-but-bounded disturbances; (iii) a single-loop state-feedback control strategy for trajectory tracking of the suspended load, robust to unmodeled dynamics, parametric uncertainties and external disturbances, with enlarged domain of attraction; and (iv) formulation of pole placement constraints in discrete-time for overshoot requirements.

This paper is organized as follows: the dynamic equations of the tilt-rotor UAV with suspended load are developed in Section \ref{sec:modeling}, from the perspective of the load; Section \ref{sec:estimation} proposes the zonotopic state estimation strategy to provide the entire state vector, formulated for sensors with different sampling times, and also the derivation of a Kalman filter is presented for comparison purposes; Section \ref{sec:control} presents the design of the state-feedback mixed $\mathcal{H}_2/\mathcal{H}_\infty$ control strategy with constraints in pole placement for path tracking of the suspended load, with feedback from estimated states; Section \ref{sec:results} presents results from numerical experiments to demonstrate and compare the performance of the zonotopic state estimator along with the designed controller; and Section \ref{sec:conclusions} concludes the work.

\section{System modeling from the perspective of the load} \label{sec:modeling}

This section develops the equations of motion of the tilt-rotor UAV with suspended load, formulated from the perspective of the load. The system is regarded as a multi-body mechanical system, and its dynamic equations are obtained through the Euler-Lagrange formulation. The dynamic coupling between the aircraft and the load is taken into account naturally. By choosing the latter's position and orientation as degrees of freedom, nonlinear state-space equations are obtained with these coordinates represented by state variables. The aircraft's position and orientation are described only with respect to the load.

\subsection{System description} \label{sec:modeling_system}

The tilt-rotor UAV with suspended load is regarded as a multi-body mechanical system composed of four rigid bodies: (i) the aircraft's main body, composed of Acrylonitrile Butadiene Styrene (ABS) structure, landing skids, batteries, instrumentation and electronics; (ii) the right thruster group, composed of the right thruster and a tilting mechanism (a revolute joint); (iii) the left thruster group, composed of the left thruster and a tilting mechanism; and (iv) the suspended load group, composed of the load and the rope. The actuators of the system are the aircraft's thrusters and servomotors. The Computer Aided Design (CAD) model of the tilt-rotor UAV with suspended load is shown in Figure \ref{fig:modeling_tiltrotor}. 
\begin{figure}[!htbp]
	\centering{
		\includegraphics[width = 0.35\textwidth]{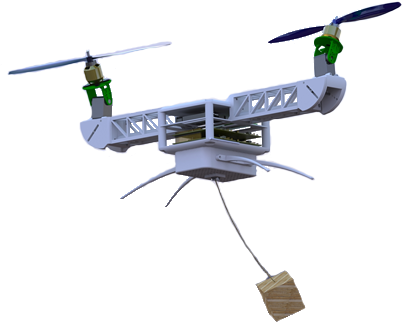}
		\caption{The tilt-rotor UAV with suspended load (CAD model).}\label{fig:modeling_tiltrotor}}
\end{figure}

For control purposes, the rope is assumed to be rigid and massless. Moreover, the aircraft's center of mass is displaced from its geometric center in order to improve pitch moment and to yield non-null equilibria for the angular positions of the tilting mechanisms and pitch angle. This mechanical feature improves the controllability of the aircraft in hover flight, yielding horizontal projections of the thrust forces without tilting the thrusters.

\subsection{Kinematics from the perspective of the load}

The approach presented in this paper consists in formulating the forward kinematics of the system considering the load as a free rigid body, while the tilt-rotor UAV as a multi-link system rigidly coupled to it. For such objective, six reference frames are defined, shown in Figure \ref{fig:modeling_kinematics}: (i) the inertial reference frame, $\frI$; (ii) the suspended load group center of mass frame, $\frL$; (iii) the aircraft's geometric center frame, $\frB$; (iv) the main body center of mass frame, $\frC{1}$; (v) the right thruster group center of mass frame, $\frC{2}$; and (vi) the left thruster group center of mass frame, $\frC{3}$. Three auxiliary frames are also defined: (i) a reference frame attached to the point of connection of the rope to the aircraft, $\frA{1}$; (ii) a reference frame attached to the tilting axis of the right servomotor, $\frA{2}$; and (iii) a reference frame attached to the tilting axis of the left servomotor, $\frA{3}$. 
\begin{figure}[ht]
	\centering{
		\def\svgwidth{0.6\textwidth}
		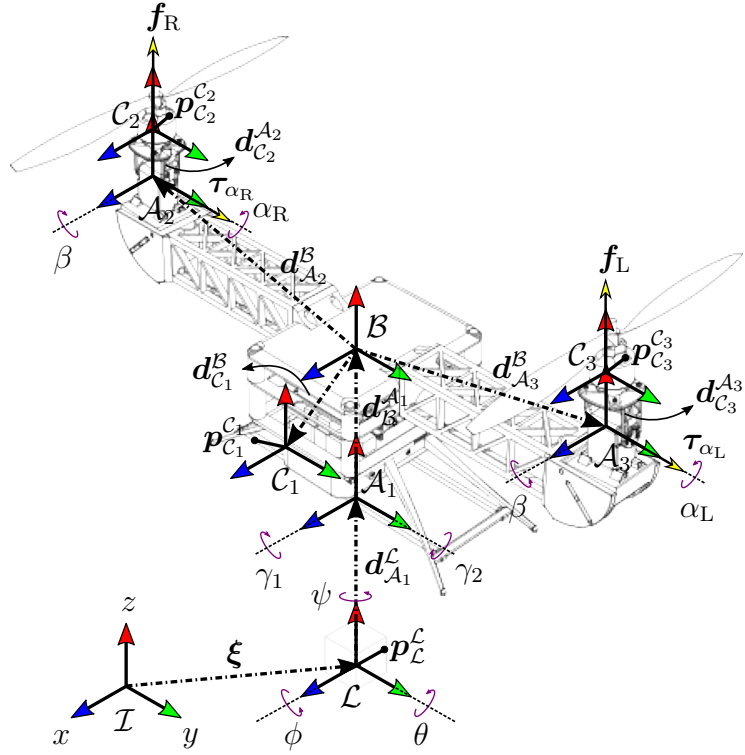
		\caption{Kinematic definitions, input forces and torques.}\label{fig:modeling_kinematics}}
\end{figure}

The position of the load with respect to the inertial frame $\frI$ is denoted by $\bm{\xi} \triangleq [ x \;\, y \;\, z ]^T$. The displacement vector from $\frL$ to $\frA{1}$ corresponds to the rope, and is defined in $\frL$ by $\dLA{1} \triangleq [ 0 \;\, 0 \;\, l ]^T$, where $l$ is the rope's length. The displacement vectors from $\frA{1}$ to $\frB$, from $\frB$ to $\frC{1}$, from $\frB$ to $\frA{i}$, and from $\frA{i}$ to $\frC{i}$ are model parameters of the tilt-rotor UAV, denoted by $\dAB{1}$, $\dBC{1}$, $\dBA{i}$, $\dAC{i}{i}$, respectively, expressed in the respective previous frames, with $i \in \{2,3\}$.

The orientation of the load with respect to $\frI$ is parametrized by Euler angles, $\bm{\eta} \triangleq [ \phi \; \theta \; \psi ]^T$, using the $ZYX$ convention about local axes. The associated rotation matrix is defined by
\begin{equation}
	\RIL \triangleq \bm{R}_{z,\psi} \bm{R}_{y,\theta} \bm{R}_{x,\phi} =  \begin{bmatrix}
		\text{c}_\psi \text{c}_\theta & \text{c}_\psi \text{s}_\theta \text{s}_\phi - \text{s}_\psi \text{c}_\phi & \text{c}_\psi \text{s}_\theta \text{c}_\phi + \text{s}_\psi \text{s}_\phi \\
		\text{s}_\psi \text{c}_\theta & \text{s}_\psi \text{s}_\theta \text{s}_\phi + \text{c}_\psi \text{c}_\phi & \text{s}_\psi \text{s}_\theta \text{c}_\phi - \text{c}_\psi \text{s}_\phi \\
		-\text{s}_\theta & \text{c}_\theta \text{s}_\phi & \text{c}_\theta \text{c}_\phi
	\end{bmatrix} \text{.}
\end{equation}

Since the rope is assumed rigid, it cannot twist. Thus, the orientation of frame $\frA{1}$ with respect to $\frL$, corresponding to the orientation of the UAV with respect to the rope, is parametrized by two angles, $\bm{\gamma} \triangleq [ \gone \; \gtwo ]^T$, such that
\begin{equation}
	\RLA{1} \triangleq \bm{R}_{x,-\gone} \bm{R}_{y,-\gtwo} = \begin{bmatrix}
		\text{c}_\gtwo & 0 & -\text{s}_\gtwo \\
		\text{s}_\gone \text{s}_\gtwo & \text{c}_\gone & \text{s}_\gone \text{c}_\gtwo \\
		\text{c}_\gone \text{s}_\gtwo & -\text{s}_\gone & \text{c}_\gone \text{c}_\gtwo
	\end{bmatrix} \text{.}
\end{equation}

The orientations of the thrusters' groups with respect to $\frB$ are defined by
\begin{equation}
	\RBA{2} {\triangleq} \bm{R}_{x,-\beta} \bm{R}_{y,\aR} {=} \begin{bmatrix}
		\text{c}_\aR & 0 & \text{s}_\aR \\
		-\text{s}_\beta \text{s}_\aR & \text{c}_\beta & \text{s}_\beta \text{c}_\aR \\
		-\text{c}_\beta \text{s}_\aR & -\text{s}_\beta & \text{c}_\beta \text{c}_\aR
	\end{bmatrix} \! \text{,} ~
	\RBA{3} {\triangleq} \bm{R}_{x,\beta} \bm{R}_{y,\aL} {=} \begin{bmatrix}
		\text{c}_\aL & 0 & \text{s}_\aL \\
		\text{s}_\beta \text{s}_\aL  & \text{c}_\beta & - \text{s}_\beta \text{c}_\aL \\
		-\text{c}_\beta \text{s}_\aL & \text{s}_\beta &   \text{c}_\beta \text{c}_\aL
	\end{bmatrix} \text{,}
\end{equation}
where $\aR$ and $\aL$ are the tilting angles of the right and left servomotors, respectively, and $\beta$ is a fixed inclination angle of the thrusters towards the aircraft geometric center, designed to improve the aircraft controllability \citep{Raffo2011}. %
The reference frames $\frA{1}$, $\frB$, and $\frC{1}$ are parallel to each other and attached to the same rigid body, thus the relative orientation is null, i.e., $\RAB{1} = \RBC{1} = \eye{3}$\footnote{In this work, $\eye{n}$ is an identity matrix with dimension $n$, $\zeros{n}{m}$ denotes an $n$ by $m$ matrix of zeros, and $\ones{n}{m}$ is an $n$ by $m$ matrix of ones.}. Similarly, we have that $\RAC{2}{2} = \RAC{3}{3} = \eye{3}$. Then, $\RLB \triangleq \RLA{1} \RAB{1} = \RLA{1}$, and $\RBC{i} \triangleq \RBA{i} \RAC{i}{i} = \RBA{i}$, for $i \in \{2,3\}$.

Given that $\dot{\RIL} = \RIL S(\wLIL)$, where $\wLIL$ denotes the angular velocity of $\frL$ with respect to $\frI$, expressed in $\frL$, and $S(\cdot)$ denotes an operator that maps a vector to a skew-symmetric matrix \citep{Siciliano2009}, we have $\wLIL = \Weta \deta$. Similarly, the remainder angular velocities of the system are given by $\wALA{1} = \bm{Q} \dgamma$, $\wBAB{1} = \wiBi{1} = \wiAi{2} = \wiAi{3} = \zeros{3}{1}$, $\wABA{2} = \ay \daR$, and $\wABA{3} = \ay \daL$,  with
\begin{equation}
	\label{eq:WeQ}
	\Weta \triangleq \begin{bmatrix}
		1 & 0 & -\text{s}_\theta \\
		0 & \text{c}_\phi & \text{s}_\phi \text{c}_\theta \\
		0 & -\text{s}_\phi & \text{c}_\phi \text{c}_\theta
	\end{bmatrix}\text{,} \quad \quad \bm{Q} \triangleq \begin{bmatrix}
	-\text{c}_\gtwo & 0 \\
	0 & -1 \\
	\text{s}_\gtwo & 0
\end{bmatrix}\text{,}
\end{equation}
and $\ay \triangleq [ 0 \;\, 1 \;\, 0 ]^T$. Moreover, $\wBLB \triangleq \bm{\omega}^\frB_{\frL \frA{1}} + \wBAB{1} = (\RAB{1})^T \wALA{1}  + \wBAB{1} = \wALA{1}$, $\wiBi{2} \triangleq \bm{\omega}^\frC{2}_{\frB \frA{2}} + \wiAi{2} = (\RAC{2}{2})^T \wABA{2}  + \wiAi{2} = \wABA{2}$, and $\wiBi{3} \triangleq \bm{\omega}^\frC{3}_{\frB \frA{3}} + \wiAi{3} = (\RAC{3}{3})^T \wABA{3}  + \wiAi{3} = \wABA{3}$. 

The defined rigid transformations yield the forward kinematics of points that belong to each rigid body, given by
\begin{align}
	& \pIL = \bm{\xi} + \RIL \pLL \text{,} \label{eq:pIL}\\
	& \pIi{1} = \bm{\xi} + \RIL \dLA{1} + \RIL \RLB (\dAC{1}{1}) + \RIL \RLB \RBC{1} \pii{1} \text{,}\label{eq:pI1} \\
	& \pIi{i} = \bm{\xi} + \RIL \dLA{1} + \RIL \RLB (\dAA{1}{i}) + \RIL \RLB \RBC{i} (\dAC{i}{i} + \pii{i}) \text{,}\label{eq:pIi}
\end{align}
where $\dAC{1}{1} \triangleq \dAB{1} + \dBC{1}$ and $\dAA{1}{i} \triangleq \dAB{1} + \dBA{i}$, $\bm{p}_\frL$ is the position of a point that belongs to the suspended load body, and $\bm{p}_\frC{i}$ belongs to the rigid body with attached frame $\frC{i}$.

The generalized coordinates of the system are chosen according to the defined rigid transformations. \textred{Note that, since the load's position and orientation are defined with respect to $\frI$, such variables are independent of each other. Therefore, these} are included in the generalized coordinates vector, which is chosen as
\begin{equation} \label{eq:generalizedcoordinates}
	\bm{q} \triangleq \begin{bmatrix} \bm{\xi}^T & \bm{\eta}^T & \bm{\gamma}^T & \aR & \aL \end{bmatrix}^T \in \setreal^{10}\text{.}
\end{equation}

Due to the chosen perspective, the position and orientation of the aircraft with respect to $\frI$ are not degrees of freedom of the system, being not included in \eqref{eq:generalizedcoordinates}. Consequently, their time evolution will not be described explicitly by the obtained equations of motion.

\subsection{\textred{Equations of motion}}

In order to derive the equations of motion through the Euler-Lagrange formulation, on one hand the kinetic and potential energies of each body of the mechanical system must be obtained. These energies can be computed for the $j$-th rigid body through the volume integrals \citep{Siciliano2009}
\begin{gather}
	\Kin_j = \frac{1}{2} \int_{V_j} \rho_j (\dot{\bm{p}}_j^\frI)^T (\dot{\bm{p}}_j^\frI) \text{d}V_j  \text{,} \label{eq:kineticenergydefinition}\\ 
	\Pot_j = - \int_{V_j} \rho_j \hat{\bm{g}}^T \pIi{j} \text{d}V_j = - m_j \hat{\bm{g}}^T \oIi{j} \text{,}  \label{eq:potentialenergydefinition}
\end{gather}
respectively, where $\rho_j$ denotes its density, $V_j$ its volume, $m_j \triangleq \int_{V_j} \rho_j \text{d}V_j$ its mass, $\hat{\bm{g}} \triangleq [0 \; 0 \; -\hat{g}_z]^T$ corresponds to the gravitational acceleration vector expressed in $\frI$, and $\oIi{j}$ is the position vector obtained from the forward kinematics of the origin of $\frC{j}$. The quadratic terms $(\dpIL)^T \dpIL$, $(\dpIi{1})^T \dpIi{1}$, and $(\dpIi{i})^T \dpIi{i}$ are computed using the time derivatives of (\ref{eq:pIL})--(\ref{eq:pIi}), respectively. Moreover, by defining the inertia tensors $\IL = \int_{V_\frL} \rho_\frL S(\pLL)^TS(\pLL) \text{d}V_\frL$ and $\Ii{j} = \int_{V_j} \rho_j S(\pii{j})^T S(\pii{j}) \text{d}V_j$, taking into account the parallel axis theorem \citep{Shabana2010}, yields $\Ji{j} \triangleq \mi{j} S(\dAC{j}{j})^TS(\dAC{j}{j}) + \RAC{j}{j} \Ii{j} (\RAC{j}{j})^T$, $j \in \{1,2,3\}$, $\Ei{i} \triangleq \mi{i} S(\dAA{1}{i})^TS(\dAA{1}{i}) + \RAA{1}{i} \Ii{i} (\RAA{1}{i})^T$, $\Di{1} \triangleq \mi{1} S(\dLA{1})^TS(\dLA{1}) + \RLA{1} \Ji{1} (\RLA{1})^{T}$, and $\Di{i} \triangleq \mi{i} S(\dLA{1})^TS(\dLA{1}) + \RLA{1} \Ei{i} (\RLA{1})^{T}$, $i \in \{2,3\}$.

The total kinetic energy of the system is given by $\Kin = \KinL + \sum_{j=1}^3 \Kin_j$, in which the kinetic energy of each rigid body is obtained using \eqref{eq:kineticenergydefinition}. Then, by defining $\bm{\Phi}_i \triangleq \RBC{i} S(m_i \dAC{i}{i})^T (\RBC{i})^T$, $\tilde{\bm{\Phi}}_i \triangleq \RBC{i} S(m_i \dAC{i}{i})^T$, $\bm{\Theta}_i \triangleq S(\dAA{1}{i}) \bm{\Phi}_i$, and $\tilde{\bm{\Theta}}_i \triangleq S(\dAA{1}{i}) \tilde{\bm{\Phi}}_i$, and using several properties of skew-symmetric matrices \citep{Siciliano2009}, writing the total kinetic energy as $\Kin = \half \dot{\bm{q}}^T \bm{M}(\bm{q}) \dot{\bm{q}}$ leads to the inertia matrix $\bm{M}(\bm{q}) \in \setrealmat{10}{10}$,
\begin{equation} \label{eq:modeling_inertiamatrix}
	\bm{M}(\bm{q}) = \begin{bmatrix}
		(\mL + m) \eye{3}   & \bm{M}_{12}   & \bm{M}_{13}   & \RIL\RLB\tilde{\bm{\Phi}}_2\ay            & \RIL\RLB\tilde{\bm{\Phi}}_3\ay \\
		*                   & \bm{M}_{22}   & \bm{M}_{23}   & \bm{M}_{24}           & \bm{M}_{25} \\
		*                   & *             & \bm{M}_{33}    & \bm{Q}^T [\RBC{2} \Ji{2} + \tilde{\bm{\Theta}}_2] \ay   & \bm{Q}^T [\RBC{3} \Ji{3} + \tilde{\bm{\Theta}}_3] \ay \\
		*                   & *             & *                         & \ay^T \Ji{2} \ay              & 0 \\
		*                   & *             & *                         & *                             & \ay^T \Ji{3} \ay
	\end{bmatrix} \text{,}
\end{equation}
where $*$ denotes elements that are deduced by symmetry, and
\begin{align}
	 \bm{M}_{12} & = -m \RIL S(\dLA{1}) \Weta - \RIL \RLB S(\dm) (\RLB)^T \Weta + \RLB \RLB \bm{\Phi} (\RLB)^T \Weta \text{,}\\
	 \bm{M}_{13} & = -\RIL \RLB S(\dm) \bm{Q} + \RIL \RLB \bm{\Phi} \bm{Q} \text{,}\\
	 \bm{M}_{22} & = \Weta^T \left[ \IL + \bm{D} - S(\dLA{1}) \RLB S(\dm) (\RLB)^T - \RLB S(\dm) (\RLB)^T S(\dLA{1}) \right. \\
	 & \quad \left. + S(\dLA{1}) \RLB \bm{\Phi} (\RLB)^T - \RLB \bm{\Phi}^T (\RLB)^T S(\dLA{1}) + \RLB(\bm{\Theta} + \bm{\Theta}^T)(\RLB)^T \right] \Weta \text{,} \nonumber\\
	 \bm{M}_{23} & = \Weta^T \left[ - S(\dLA{1}) \RLB S(\dm) + \RLB (\Ji{1} + \bm{E}) + S(\dLA{1}) \RLB \bm{\Phi} + \RLB(\bm{\Theta} + \bm{\Theta}^T) \right] \bm{Q} \text{,}\\
	 \bm{M}_{24} & = \Weta^T \left[ \RLB\RBC{2}\Ji{2} + S(\dLA{1}) \RLB \tilde{\bm{\Phi}}_2 + \RLB \tilde{\bm{\Theta}}_2  \right] \ay \text{,} \\
	 \bm{M}_{25} & = \Weta^T \left[ \RLB\RBC{3}\Ji{3} + S(\dLA{1}) \RLB \tilde{\bm{\Phi}}_3 + \RLB \tilde{\bm{\Theta}}_3  \right] \ay \text{,}\\
	 \bm{M}_{33} & = \bm{Q}^T [\Ji{1} + \bm{E} + \bm{\Theta} + \bm{\Theta}^T] \bm{Q} \text{,}
\end{align}
with $m \triangleq \sum_{j=1}^{3} \mi{j}$, $\bm{E} \triangleq \sum_{i=2}^{3} \Ei{i}$, $\bm{D} \triangleq \sum_{j=1}^{3} \Di{j}$, $\bm{\Phi} \triangleq \sum_{i=2}^{3} \bm{\Phi}_i$, $\bm{\Theta} \triangleq \sum_{i=2}^{3} \bm{\Theta}_i$, and $\dm \triangleq m_1 \dAC{1}{1} + \sum_{i=2}^{3} \mi{i} \dAA{1}{i}$. Note from the inertia matrix that the dynamics of the four rigid bodies are coupled, allowing one to consider the existing interactions in the control design without the need of cascade control structures.

The Coriolis and centripetal forces matrix, $\bm{C}(\bm{q},\dot{\bm{q}}) \in \setrealmat{10}{10}$, is obtained through Christoffel symbols of the first kind \citep{Siciliano2009}. The element of its $k$-th row and $j$-th column is given by
\begin{equation} \label{eq:modeling_coriolismatrix}
	C_{kj} = \sum_{i=1}^{10} \dfrac{1}{2} \left( \dfrac{\partial M_{kj}}{\partial q_i} + \dfrac{\partial M_{ki}}{\partial q_j} - \dfrac{\partial M_{ij}}{\partial q_k} \right) \dot{q}_i \text{,}
\end{equation}
where $k,j \in \{1,2,\dots,10\}$, with $M$ being an element of the inertia matrix \eqref{eq:modeling_inertiamatrix}.

The forward kinematics of each body's center of mass is obtained using \eqref{eq:pIL}--\eqref{eq:pIi}, 
%
\textred{with the potential energies of the load and of each body of the aircraft then computed using \eqref{eq:potentialenergydefinition}}. The total potential energy of the system is given by $\Pot = \PotL + \sum_{j=1}^{3} \Poti{j}$, yielding 
\begin{equation}
	\Pot = - \hat{\bm{g}}^T \left[ (\mL + m) \bm{\xi} + m \RIL \dLA{1} + \RIL \RLB \dm + \sum_{i=2}^3 m_i \RIL\RLB\RBC{i} \dAC{i}{i}\right] \text{.}
\end{equation}

The gravitational force vector is then obtained through
\begin{equation} \label{eq:modeling_gravitationalforcesvector}
	\bm{g}(\bm{q}) = \frac{\partial \Pot}{\partial \bm{q}} \in \setreal^{10} \text{.}
\end{equation}



\textred{On the other hand, this work assumes that the system is also subject to generalized forces from the aircraft's actuators, viscous friction, and external disturbances affecting the load.} Therefore, let $\bm{f} \in \setreal^{3}$ and $\bm{\tau} \in \setreal^{3}$ denote non-conservative force and torque vectors, respectively, actuating on the mechanical system. Furthermore, let $\bm{p} \in \setreal^{3}$ denote the point of application of $\bm{f}$, and $\mathcal{F}$ be a reference frame rigidly attached to the body to which $\bm{\tau}$ is applied. \textred{The contributions of $\bm{f}$ and $\bm{\tau}$ to the generalized forces can be computed through \cite{Kane1985}}
\begin{align} 
	& \gforce_{\bm{f}} = (\Jac_{\bm{p}})^T \bm{f}^\frI \in \setreal^{n}\text{,} \label{eq:gforcebyforce}\\
	& \gforce_{\bm{\tau}} = (\GJac_{\mathcal{F}})^T \bm{\tau}^\frI \in \setreal^{n} \text{,} \label{eq:gforcebytorque}
\end{align}
where $\Jac_{\bm{p}} \triangleq \partial \dot{\bm{p}}^\frI / \partial \dot{\bm{q}} \in \setrealmat{3}{n}$, and $\GJac_{\mathcal{F}} \triangleq \partial \bm{\omega}^\frI_{\frI \mathcal{F}} / \partial \dot{\bm{q}} \in \setrealmat{3}{n}$.


\textred{The input forces and torques of the system are the right and left thrust forces, denoted by $\fR$ and $\fL$, and right and left servomotor torques, denoted by $\tauaR$ and $\tauaL$, respectively.} They are expressed in their respective frames by $\bm{f}^\frC{2}_\text{R} = \az \infR$, $\bm{f}^\frC{3}_\text{L} = \az \infL$, $\bm{\tau}^\frC{2}_\aR = \ay \intauaR$, and $\bm{\tau}^\frC{3}_\aL = \ay \intauaL$, where $\az \triangleq [ 0 \,\; 0 \,\; 1 ]^T$  (see Figure \ref{fig:modeling_kinematics}). \textred{Therefore}, in the inertial reference frame, we have 
\begin{align}
	& \bm{f}^\frI_\text{R} = \bm{R}^\frI_\frC{2} \bm{f}^\frC{2}_\text{R} = \RIL \RLB \RBC{2} \az \infR \text{,} \label{eq:fIR} \\
	& \bm{f}^\frI_\text{L} = \bm{R}^\frI_\frC{3} \bm{f}^\frC{3}_\text{L} = \RIL \RLB \RBC{3} \az \infL \text{,} \label{eq:fIL} \\
	& \bm{\tau}^\frI_\aR = \bm{R}^\frI_\frC{2} \bm{\tau}^\frC{2}_\aR = \RIL \RLB \RBC{2} \ay \intauaR \text{,} \label{eq:tauIaR} \\
	& \bm{\tau}^\frI_\aL = \bm{R}^\frI_\frC{3} \bm{\tau}^\frC{3}_\aL = \RIL \RLB \RBC{3} \ay \intauaL \text{.} \label{eq:tauIaL}
\end{align}

Besides, this work assumes that the thrust forces are applied to the centers of mass of the respective thrusters' groups, i.e, the origins of $\frC{2}$ and $\frC{3}$. To obtain the corresponding mappings to generalized forces, it is necessary to compute $\Jac_{\bm{o}_\frC{2}} = \partial \doIi{2}/\partial \dot{\bm{q}}$ and  $\Jac_{\bm{o}_\frC{3}} = \partial \doIi{3}/\partial \dot{\bm{q}}$. Firstly, making $\pii{2} = \bm{o}^\frC{2}_\frC{2} = \zeros{3}{1}$ in the time derivative of \eqref{eq:pIi} leads to
\begin{equation} 
	\begin{aligned}
	\doIi{2} & = \dxi {+} \left( \RIL S(\dLA{1})^T {+} \RIL \RLB S(\dAA{1}{2})^T (\RLB)^T {+} \RIL \RLB \RBC{2} S(\dAC{2}{2})^T (\RLB \RBC{2})^T \right) \Weta \deta \\
	& + \left( \RIL \RLB S(\dAA{1}{2})^T + \RIL \RLB \RBC{2} S(\dAC{2}{2})^T (\RBC{2})^T \right) \bm{Q} \dgamma + \RIL \RLB \RBC{2} S(\dAC{2}{2})^T \ay \daR \text{.} \end{aligned}
\end{equation}

Then, using \eqref{eq:gforcebyforce} and \eqref{eq:fIR} yields
\begin{align} \label{eq:gforcefR}
	\gforce_{\bm{f}_\text{R}} & = (\Jac_{\bm{o}_2})^T \bm{f}^\frI_\text{R} = \left(\partial \doIi{2}/\partial \dot{\bm{q}}\right)^T \bm{f}^\frI_\text{R} \nonumber\\
	& = \begin{bmatrix} \RIL \RLB \RBC{2} \az \\ \Weta^T S(\dLA{1}) \RLB \RBC{2} \az + \Weta^T \RLB S(\dAA{1}{2}) \RBC{2} \az + \Weta^T \RLB \RBC{2} S(\dAC{2}{2}) \az \\ \bm{Q}^T S(\dAA{1}{2}) \RBC{2} \az + \bm{Q}^T \RBC{2} S(\dAC{2}{2}) \az \\ \ay^T S(\dAC{2}{2}) \az \\ 0 \end{bmatrix} \infR \text{.}
\end{align}

Analogously, for the left thrust force, we have
\begin{equation} \label{eq:gforcefL}
	\gforce_{\bm{f}_\text{L}} = \begin{bmatrix} \RIL \RLB \RBC{3} \az \\ \Weta^T S(\dLA{1}) \RLB \RBC{3} \az + \Weta^T \RLB S(\dAA{1}{3}) \RBC{3} \az + \Weta^T \RLB \RBC{3} S(\dAC{3}{3}) \az \\ \bm{Q}^T S(\dAA{1}{3}) \RBC{3} \az + \bm{Q}^T \RBC{3} S(\dAC{3}{3}) \az \\ 0 \\ \ay^T S(\dAC{3}{3}) \az \end{bmatrix} \infR \text{.}
\end{equation}

The servomotor torques are applied to the respective thrusters' bodies, and opposite torques due to reaction are applied to the aircraft's main body. These torques are mapped to generalized forces through \eqref{eq:gforcebytorque}. From the addition of angular velocities \citep{Siciliano2009}, we have
\begin{align}
	& \bm{\omega}^\frI_{\frI \frB} = \bm{\omega}^\frI_{\frI \frL} + \bm{\omega}^\frI_{\frL \frB} = \RIL \wLIL {+} \RIL \RLB \wBLB
	= \begin{bmatrix} \zeros{3}{3} & \RIL \Weta & \RIL \RLB \bm{Q} & \zeros{3}{1} & \zeros{3}{1} \end{bmatrix}
	\dot{\bm{q}} \text{,}  \label{eq:GJacB} \\
	&\bm{\omega}^\frI_{\frI \frC{2}} = \bm{\omega}^\frI_{\frI \frL} + \bm{\omega}^\frI_{\frL \frB} + \bm{\omega}^\frI_{\frB \frC{2}} = \begin{bmatrix} \zeros{3}{3} & \RIL \Weta & \RIL \RLB \bm{Q} & \RIL \RLB \RBC{2} \ay & \zeros{3}{1} \end{bmatrix}
	\dot{\bm{q}} \text{,}
	\label{eq:GJac2} \\
	& \bm{\omega}^\frI_{\frI \frC{3}} = \bm{\omega}^\frI_{\frI \frL} + \bm{\omega}^\frI_{\frL \frB} + \bm{\omega}^\frI_{\frB \frC{3}} = \begin{bmatrix} \zeros{3}{3} & \RIL \Weta & \RIL \RLB \bm{Q} & \zeros{3}{1} & \RIL \RLB \RBC{3} \ay \end{bmatrix}
	\dot{\bm{q}} \label{eq:GJac3} \text{.}
\end{align}

\textred{Recalling that $\bm{\omega}^\frI_{\frI \frB} = \GJac_\frB \dot{\bm{q}}$, $\bm{\omega}^\frI_{\frI \frC{2}} = \GJac_\frC{2} \dot{\bm{q}}$ and $\bm{\omega}^\frI_{\frI \frC{3}} = \GJac_\frC{3} \dot{\bm{q}}$, leads to}
\begin{align}
	\gforce_{\bm{\tau}_\aR} & = (\GJac_\frC{2})^T \bm{\tau}^\frI_\aR + (\GJac_\frB)^T (-\bm{\tau}^\frI_\aR) \nonumber\\
	& = \left( \begin{bmatrix}
		\zeros{3}{3} \\ \Weta^T (\RIL)^T \\ \bm{Q}^T (\RIL \RLB)^T \\ \ay^T (\RIL \RLB \RBC{2})^T \\ \zeros{1}{3} \end{bmatrix} - \begin{bmatrix}
		\zeros{3}{3} \\ \Weta^T (\RIL)^T \\ \bm{Q}^T (\RIL \RLB)^T \\ \zeros{1}{3} \\ \zeros{1}{3} \end{bmatrix} \right) \RIL \RLB \RBC{2} \ay \intauaR 
	= \begin{bmatrix} \zeros{3}{1} \\ \zeros{3}{1} \\ \zeros{2}{1} \\ 1 \\ 0 \end{bmatrix} \intauaR \text{,} \label{eq:gforcetauaR}
\end{align}
\begin{align} \label{eq:gforcetauaL}
	\gforce_{\bm{\tau}_\aL} & = (\GJac_\frC{3})^T \bm{\tau}^\frI_\aL + (\GJac_\frB)^T (-\bm{\tau}^\frI_\aL) \nonumber\\ 
	& = \left( \begin{bmatrix}
		\zeros{3}{3} \\ \Weta^T (\RIL)^T \\ \bm{Q}^T (\RIL \RLB)^T \\ \zeros{1}{3} \\ \ay^T (\RIL \RLB \RBC{3})^T \end{bmatrix} - \begin{bmatrix}
		\zeros{3}{3} \\ \Weta^T (\RIL)^T \\ \bm{Q}^T (\RIL \RLB)^T \\ \zeros{1}{3} \\ \zeros{1}{3} \end{bmatrix} \right) \RIL \RLB \RBC{3} \ay \intauaL 
	= \begin{bmatrix} \zeros{3}{1} \\ \zeros{3}{1} \\ \zeros{2}{1} \\ 0 \\ 1 \end{bmatrix} \intauaL \text{.}
\end{align}

This work also takes into account drag torques generated by the propellers, which are reaction torques applied to the thrusters' bodies, due to the blades' acceleration and drag \citep{Castillo2005}. Assuming steady-state for the angular velocity of the blades, the drag torques are given in the thrusters' reference frames by
\begin{equation}
	\bm{\tau}^\frC{2}_\text{drag,R} = \lambdaR \ktaub \bm{f}^\frC{2}_\text{R}, \quad \bm{\tau}^\frC{3}_\text{drag,L} = \lambdaL \ktaub \bm{f}^\frC{3}_\text{L} \text{,}
\end{equation}
where $k_\tau$ and $b$ are parameters obtained experimentally, and $\lambdaR$ and $\lambdaL$ are given according to the direction of rotation of the corresponding propeller: if counter-clockwise, $1$; if clockwise, $-1$. In the inertial reference frame, we then have
\begin{align}
	& \bm{\tau}^\frI_\text{drag,R} = \bm{R}^\frI_\frC{2} \bm{\tau}^\frC{2}_\text{drag,R} = \lambdaR \ktaub \bm{R}^\frI_\frC{2} \bm{f}^\frC{2}_\text{R} = \lambdaR \ktaub \RIL \RLB \RBC{2} \az \infR \text{,} \label{eq:tauIdragR}\\
	& \bm{\tau}^\frI_\text{drag,L} = \bm{R}^\frI_\frC{3} \bm{\tau}^\frC{3}_\text{drag,L} = \lambdaL \ktaub \bm{R}^\frI_\frC{3} \bm{f}^\frC{3}_\text{L} = \lambdaL \ktaub \RIL \RLB \RBC{3} \az \infL \text{.} \label{eq:tauIdragL}
\end{align}

The drag torques are applied to the thrusters' bodies, then \textred{from} (\ref{eq:gforcebytorque}), (\ref{eq:GJac2}), (\ref{eq:GJac3}), (\ref{eq:tauIdragR}), and (\ref{eq:tauIdragL}), yields
\begin{equation} 
	\gforce_{\bm{\tau}_\text{drag,R}} = (\GJac_\frC{2})^T \bm{\tau}^\frI_\text{drag,R} = \lambdaR \ktaub \begin{bmatrix}
		\zeros{3}{1} \\ \Weta^T \RLB \RBC{2} \az \\ \bm{Q}^T \RBC{2} \az \\ 0 \\ 0 \end{bmatrix} \infR \text{,} \label{eq:gforcetaudragR}
\end{equation}
\begin{equation} \label{eq:gforcetaudragL}
	\gforce_{\bm{\tau}_\text{drag,L}} = (\GJac_\frC{3})^T \bm{\tau}^\frI_\text{drag,L} = \lambdaL \ktaub \begin{bmatrix}
		\zeros{3}{1} \\ \Weta^T \RLB \RBC{3} \az \\ \bm{Q}^T \RBC{3} \az \\ 0 \\ 0 \end{bmatrix} \infL \text{.}
\end{equation}

Finally, the complete mapping of the control inputs to generalized forces is obtained by summing up the contributions of the thrust forces, servomotor torques, and drag torques. Thus, from \eqref{eq:gforcefR}, \eqref{eq:gforcefL}, \eqref{eq:gforcetauaR}, \eqref{eq:gforcetauaL}, \eqref{eq:gforcetaudragR}, and \eqref{eq:gforcetaudragL}, 
\begin{align}
	\gforce_\text{in} & = \gforce_{\bm{f}_\text{R}} + \gforce_{\bm{f}_\text{L}} + \gforce_{\bm{\tau}_\aR} + \gforce_{\bm{\tau}_\aL} + \gforce_{\bm{\tau}_\text{drag,R}} + \gforce_{\bm{\tau}_\text{drag,L}} \nonumber\\
	& = \begin{bmatrix}
		\RIL \RLB \RBC{2} \az & \RIL \RLB \RBC{3} \az & \zeros{3}{1} & \zeros{3}{1} \\
		\Weta^T \LambdaR \az & \Weta^T \LambdaL \az & \zeros{3}{1} & \zeros{3}{1} \\
		\bm{Q}^T \GammaR \az & \bm{Q}^T \GammaL \az & \zeros{2}{1} & \zeros{2}{1} \\
		\ay^T S(\dAC{2}{2}) \az &                       0 & 1 & 0 \\
		                      0 & \ay^T S(\dAC{3}{3}) \az & 0 & 1	
	\end{bmatrix} \begin{bmatrix} \infR \\ \infL \\ \intauaR \\ \intauaL \end{bmatrix} \triangleq \bm{L}_\text{in}(\bm{q}) \bm{u} \label{eq:gforceinputs} \text{,}
\end{align}
where
\begin{align}
	& \LambdaR \triangleq S(\dLA{1}) \RLB \RBC{2} + \RLB S(\dAA{1}{2}) \RBC{2} + \RLB \RBC{2} S(\dAC{2}{2}) + \lambdaR \ktaub \RLB \RBC{2} \text{,} \\
	& \LambdaL \triangleq S(\dLA{1}) \RLB \RBC{3} + \RLB S(\dAA{1}{3}) \RBC{3} + \RLB \RBC{3} S(\dAC{3}{3}) + \lambdaL \ktaub \RLB \RBC{3} \text{,} \\
	& \GammaR \triangleq S(\dAA{1}{2}) \RBC{2} {+}\RBC{2} S(\dAC{2}{2}) {+} \lambdaR \ktaub \RBC{2} \text{,} \quad
	\GammaL \triangleq S(\dAA{1}{3}) \RBC{3} {+} \RBC{3} S(\dAC{3}{3}) {+} \lambdaL \ktaub \RBC{3} \text{.}
\end{align}

Although in this work no aerodynamic surfaces are considered for the tilt-rotor UAV (see Figure \ref{fig:modeling_tiltrotor}), the modeling approach presented here is general enough to describe the dynamics of any tilt-rotor carrying a suspended load, from the perspective of the latter. For such purpose, the aircraft must be regarded as a multi-body system with frame definitions similar to those shown in Figure \ref{fig:modeling_kinematics}. Moreover, if aerodynamics surfaces are considered (e.g. wings, horizontal and vertical stabilizers), the resulting lift and drag forces can be added to the model in a straightforward manner, by including the corresponding contributions in \eqref{eq:gforceinputs}, allowing the model to cope with both helicopter and airplane flight modes. 


Viscous friction is taken into account at the servomotors composing the tilting mechanisms, and also at the point of connection between the rope and the tilt-rotor UAV. The resulting friction torques are assumed to be mapped to generalized forces by
\begin{equation} \label{eq:gforcefr}
	\gforce_\text{fr} = - \bm{L}_\text{fr} \dot{\bm{q}} \text{,}
\end{equation}
where $\bm{L}_\text{fr} \triangleq \text{diag}( 0, 0, 0, 0, 0, 0, \mu_\gamma, \mu_\gamma, \mu_\alpha, \mu_\alpha)$ with $\mu_\gamma$ and $\mu_\alpha$ constant parameters.

External disturbances applied to the suspended load are also considered, which may represent wind gusts affecting the system. Defining these disturbances in the inertial reference frame as the force vector \textred{$\bm{d} \in \setreal^{3}$}, and assuming applied to the load's center of mass, it can be mapped to generalized forces through \eqref{eq:gforcebyforce}, yielding
\begin{equation} \label{eq:gforcedb}
	\gforce_\text{db} = (\partial \doIL / \partial \bm{\dot{q}})^T \bm{d} = \begin{bmatrix} \eye{3} & \zeros{3}{3} & \zeros{3}{2} & \zeros{3}{1} & \zeros{3}{1} \end{bmatrix}^T \bm{d} \triangleq \bm{L}_\text{db} \bm{d} \text{,}
\end{equation}
where $\doIL = \dxi$ is obtained by making $\pLL = \zeros{3}{1}$ in the time derivative of \eqref{eq:pIL}.


The equations of motion of the tilt-rotor UAV with suspended load can be written in the Euler-Lagrange formulation as \citep{Siciliano2009} 
\begin{equation} \label{eq:eulerlagrangeCANONICAL}
	\bm{M}(\bm{q})\ddot{\bm{q}} + \bm{C}(\bm{q},\dot{\bm{q}})\dot{\bm{q}} + \bm{g}(\bm{q}) = \gforce \text{,}
\end{equation}
where $\bm{M}(\bm{q})$, $\bm{C}(\bm{q},\dot{\bm{q}})$ and $\bm{g}(\bm{q})$ are given by \eqref{eq:modeling_inertiamatrix}, \eqref{eq:modeling_coriolismatrix}, and \eqref{eq:modeling_gravitationalforcesvector}, respectively, and  
$\gforce$ is the total generalized forces vector, obtained by $\gforce = \gforce_\text{in} + \gforce_\text{fr} + \gforce_\text{db} = \bm{L}_\text{in}(\bm{q}) \bm{u} - \bm{L}_\text{fr} \dot{\bm{q}} + \bm{L}_\text{db} \bm{d}$. 
%
%
\textred{By defining the state vector}
\begin{equation} \label{eq:modeling_systemstates}
	\bm{x} \triangleq \begin{bmatrix} \bm{q}^T & \dot{\bm{q}}^T \end{bmatrix}^T \in \setreal^{20}\text{,}
\end{equation}
\textred{the dynamic equations \eqref{eq:eulerlagrangeCANONICAL} can be written in the state-space representation}
\begin{equation} \label{eq:modeling_statespacenonlinear}
	\dot{\bm{x}} = \bm{\varphi}(\bm{x}, \bm{u}, \bm{d}) = \begin{bmatrix}
		\dot{\bm{q}} \\
		\bm{M}(\bm{q})^{-1} \left[ - (\bm{C}(\bm{q},\dot{\bm{q}})  + \bm{L}_\text{fr}) \dot{\bm{q}} - \bm{g}(\bm{q}) + \bm{L}_\text{in}(\bm{q}) \bm{u} + \bm{L}_\text{db}\bm{d} \right]
	\end{bmatrix} \text{,}
\end{equation}
which is nonlinear and highly coupled. Since the position and orientation of the load belong to the generalized coordinates \eqref{eq:generalizedcoordinates}, they are represented by the state variables \eqref{eq:modeling_systemstates}. Consequently, the load's behavior is described explicitly by \eqref{eq:modeling_statespacenonlinear}. On the other hand, the aircraft's position and orientation are described only with respect to the load, thus appearing in \eqref{eq:modeling_statespacenonlinear} only implicitly.

\section{State estimation} \label{sec:estimation}

The developed state-space representation \eqref{eq:modeling_statespacenonlinear} describes explicitly the dynamics of the load's position and orientation, which are represented by state variables. This fact allows state-feedback control strategies to directly steer the trajectory of the suspended load with respect to the inertial reference frame. However, in real applications, often the available sensors provide information only about the aircraft's position and orientation, which prevents to directly measure all the system states \eqref{eq:modeling_systemstates}.

This section proposes the design of a zonotopic state estimator (ZSE) to provide the state variables of the tilt-rotor UAV with suspended load. A realistic scenario is considered, in which the load's position and orientation with respect to the inertial reference frame are not measured. For comparison purposes, the derivation of a Kalman filter (KF) is also presented. The following sensors are assumed to be available: (i) a Global Positioning System (GPS), providing the position of the UAV with respect to the inertial reference frame $\frI$, along axes $x$ and $y$; (ii) a barometer, providing the position of the UAV with respect to $\frI$, along axis $z$; (iii) an Inertial Measurement Unit (IMU), providing the orientation and angular velocities of the UAV with respect to $\frI$, the latter expressed in the geometric center frame $\frB$; (iv) a camera, providing the load's position with respect to the point of connection of the rope, expressed in $\frA{1}$; and (v) embedded sensors at the servomotors, providing the tilting angles of the propellers and their time derivatives. The measured information is assumed to be corrupted with noise, and each sensor has its own sampling time.

\subsection{Measurement equation}

In order to design the state estimators, consider a measurement equation of the form $\bm{y}_k = \bm{\pi}(\bm{x}_k) + \bm{v}_k$,
where $\bm{y}_k$ is the measured vector at time instant $k$, $\bm{v}_k$ corresponds to measurement noise, and $\bm{\pi}(\bm{x}_k)$ is a nonlinear mapping of the system states to the measured variables. %

This work assumes that the GPS, barometer, and IMU are located at the geometric center of the aircraft, whilst the camera is located at the origin of $\frA{1}$. Let $\xiB \triangleq [x_\frB \; y_\frB \; z_\frB]^T$ denote the position of the tilt-rotor UAV with respect to $\frI$, as shown in Figure \ref{fig:estimation_measuredvariables}. Then, by forward kinematics, we have that
\begin{equation} \label{eq:estimation_gpsbarometerequation}
	\xiB(\bm{\xi},\bm{\eta}) = \bm{\xi} + \RIL \dLA{1} + \RIL \RLB \dAB{1} \text{.} 
\end{equation}
\begin{figure}[ht]
	\begin{footnotesize}
		\centering{
			\def\svgwidth{0.32\textwidth}
			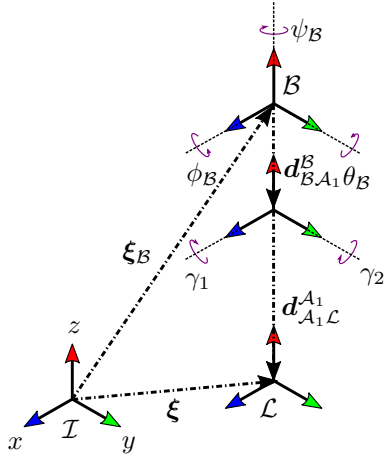
			\caption{Measured position and orientation of the tilt-rotor UAV, and measured position of the load.}
			\label{fig:estimation_measuredvariables}}
	\end{footnotesize}
\end{figure}

The orientation of the aircraft with respect to $\frI$ is assumed to be parametrized by Euler angles, denoted by $\etaB \triangleq [\phiB \; \theta_\frB \; \psi_\frB]^T$, using the local roll-pitch-yaw convention. Therefore,
\begin{equation}
	\RIB \triangleq \bm{R}_{z,\psiB} \bm{R}_{y,\thetaB} \bm{R}_{x,\phiB} =  \begin{bmatrix}
		\text{c}_{\psiB} \text{c}_\thetaB & \text{c}_\psiB \text{s}_\thetaB \text{s}_\phiB - \text{s}_\psiB \text{c}_\phiB & \text{c}_\psiB \text{s}_\thetaB \text{c}_{\phiB} + \text{s}_\psiB \text{s}_\phiB \\
		\text{s}_\psiB \text{c}_\thetaB & \text{s}_\psiB \text{s}_\thetaB \text{s}_\phiB + \text{c}_\psiB \text{c}_\phiB & \text{s}_\psiB \text{s}_\thetaB \text{c}_\phiB - \text{c}_\psiB \text{s}_\phiB \\
		-\text{s}_\thetaB & \text{c}_\thetaB \text{s}_\phiB & \text{c}_\thetaB \text{c}_\phiB
	\end{bmatrix} \text{,}
\end{equation}
from which, %
%
since $\RIB = \RIL \RLB$, the following holds
\begin{gather}
	\phiB(\bm{\eta},\bm{\gamma}) = \arctan\left(\frac{{(\RIL \RLB)}_{32}}{{(\RIL \RLB)}_{33}}\right) \text{,} \label{eq:estimation_IMUetaBequation01}\\
	\thetaB(\bm{\eta},\bm{\gamma}) = \arcsin\left(-{(\RIL \RLB)}_{31}\right) \text{,} \label{eq:estimation_IMUetaBequation02}\\
	\psiB(\bm{\eta},\bm{\gamma}) = \arctan\left(\frac{{(\RIL \RLB)}_{21}}{{(\RIL \RLB)}_{11}}\right) \text{,} \label{eq:estimation_IMUetaBequation03}
\end{gather}
for $\thetaB \neq \pm \pi/2$, where $(\cdot)_{ij}$ denotes the element from the $i$-th line and $j$-th column.

The angular velocity provided by the IMU is given by
\begin{equation}
	\wBIB(\bm{\eta},\bm{\gamma},\deta,\dgamma) = \bm{\omega}^\frB_{\frI \frL} + \wBLB = (\RLB)^T \Weta \deta + \bm{Q} \dgamma \text{,} \label{eq:estimation_IMUwBIBequation}
\end{equation}
where $\Weta$ and $\bm{Q}$ were defined in \eqref{eq:WeQ}.

Let $\bm{d}^\frA{1}_{\frA{1}\frL}$ denote the measurement provided by the camera, which correspond to the displacement vector from $\frA{1}$ to $\frL$, expressed in $\frA{1}$ (see Figure \ref{fig:estimation_measuredvariables}). Therefore,
\begin{equation} \label{eq:estimation_cameraequation}
	\bm{d}^\frA{1}_{\frA{1}\frL}(\bm{\gamma}) = - \bm{d}^\frA{1}_{\frL\frA{1}} = - (\RLA{1})^T \bm{d}^\frL_{\frL\frA{1}} = - (\RLB)^T \bm{d}^\frL_{\frL\frA{1}}\triangleq - (\RLB)^T \bm{d}^\frL_{\frA{1}} \text{.}
\end{equation}

Gathering equations \eqref{eq:estimation_gpsbarometerequation} through \eqref{eq:estimation_cameraequation} along with the system states $\aR$, $\aL$, $\daR$ and $\daL$, measured by the sensors at the servomotors, and considering the measurement noise $\bm{v}_k$, yields the nonlinear measurement equation
\begin{equation} \label{eq:estimation_measurementequations}
	\bm{y}_k = \bm{\pi}(\bm{x}_k) + \bm{v}_k \triangleq \begin{bmatrix} \xiB(\bm{\xi},\bm{\eta}) \\ \phiB(\bm{\eta},\bm{\gamma}) \\ \thetaB(\bm{\eta},\bm{\gamma}) \\ \psiB(\bm{\eta},\bm{\gamma}) \\ \wBIB(\bm{\eta},\bm{\gamma},\deta,\dgamma) \\ \bm{d}^\frA{1}_{\frA{1}\frL}(\bm{\gamma}) \\ \aR \\ \aL \\ \daR \\ \daL \end{bmatrix} + \bm{v}_k = \begin{bmatrix}
		\bm{\xi} + \RIL \dLA{1} + \RIL \RLB \dAB{1} \\ \arctan\left({(\RIL \RLB)}_{32}/{(\RIL \RLB)}_{33}\right) \\ \arcsin\left(-{(\RIL \RLB)}_{31}\right) \\
		\arctan\left({(\RIL \RLB)}_{21}/{(\RIL \RLB)}_{11}\right) \\ (\RLB)^T \Weta \deta + \bm{Q} \dgamma \\ - (\RLB)^T \dLA{1} \\ \aR \\ \aL \\ \daR \\ \daL \end{bmatrix} + \bm{v}_k \text{.}
\end{equation}
Since the sensors have different sampling times, the dimension of $\bm{y}_k$ and $\bm{v}_k$ actually changes for each $k$. Let $\setI \triangleq \{1,2,\dots,16\}$,  $\setI_k$ denote the set of available measurements at time instant $k$, such that $\setI_k \subseteq \setI$. Moreover, let $\iota_k$ denote the number of elements of $\setI_k$, then $\bm{y}_k, \bm{v}_k \in \setreal^{\iota_k}$. Assuming that at least one measurement is available for each $k$, we have that $1 \leq \iota_k \leq 16$. Then, the time-switching nonlinear mapping is defined as
\begin{equation} \label{eq:estimation_switchedmapping}
	\bm{\pi}^{[k]}(\bm{x}_k) = \left[ \bm{\pi}(\bm{x}_k)(i) \right]_{i \in \setI_k} \text{,}
\end{equation}
where $\bm{\pi}(\bm{x}_k)(i)$ denotes the $i$-th line of $\bm{\pi}(\bm{x}_k)$, and the brackets denote vertical concatenation.

\subsection{Linearized dynamic equations} \label{sec:estimation_linearizedmodel}

Initially, the state-space equations \eqref{eq:modeling_statespacenonlinear} are linearized around a time-varying trajectory. Before proceeding, some facts must be pointed out. It is possible to verify that the inertia matrix \eqref{eq:modeling_inertiamatrix} and the mapping matrix \eqref{eq:gforceinputs} are not functions of $\bm{\xi}$, and that the Coriolis matrix obtained using \eqref{eq:modeling_coriolismatrix} is neither a function of $\bm{\xi}$ nor of $\dxi$. Moreover, assuming constant gravitational acceleration, the gravitational forces vector \eqref{eq:modeling_gravitationalforcesvector} is not a function of $\bm{\xi}$\footnote{In fact, since the gravitational acceleration is assumed constant, and effects due to aerodynamic forces and air friction are being neglected, it is expected that the system dynamics are independent of its position and velocity with respect to the inertial frame.}. Thus, by defining $\bm{\zeta} \triangleq [ \bm{\eta}^T \; \bm{\gamma}^T \; \aR \; \aL]^T \in \setreal^{7}$, we have that $\bm{M} \triangleq \bm{M}(\bm{\zeta})$, $\bm{C} \triangleq \bm{C}(\bm{\zeta}, \dot{\bm{\zeta}})$, $\bm{g} \triangleq \bm{g}(\bm{\zeta})$, $\bm{L}_\text{in} \triangleq \bm{L}_\text{in}(\bm{\zeta})$. Let $\bm{x}^\text{tr}(t)$ and $\bm{u}^\text{tr}(t)$ denote trajectory values for $\bm{x}$ and $\bm{u}$, respectively. This work assumes that the desired trajectory is feasible for a disturbance-free scenario, i.e., 
\begin{equation} \label{eq:estimation_feasibletrajectory}
	\dot{\bm{x}}^\text{tr} = \bm{\varphi}(\bm{x}^\text{tr}, \bm{u}^\text{tr}, \zeros{3}{1})\text{.}
\end{equation}

Linearizing the state-space equations \eqref{eq:modeling_statespacenonlinear} around these variables, through first-order expansion in Taylor series, yields the time-varying system
\begin{gather} \label{eq:estimation_statespacelinear}
	\delta \dot{\bm{x}} = \bm{A}_\text{c}(t) \delta \bm{x} + \bm{B}_\text{c}(t) \delta \bm{u} + \bm{F}_\text{c}(t) \bm{d} \text{,}
\end{gather}
with `c' denoting continuous-time, $\delta \bm{x} \triangleq \bm{x} - \bm{x}^\text{tr}$, $\delta \bm{u} \triangleq \bm{u} - \bm{u}^\text{tr}$, and
\begin{align}
& \bm{A}_\text{c}(t) = \left.\frac{\partial \bm{\varphi}(\bm{x}, \bm{u}, \bm{d})}{\partial \bm{x}} \right|\!\!{\begin{smallmatrix}\bm{x}=\bm{x}^{\text{tr}}(t)\\ \bm{u}=\bm{u}^{\text{tr}}(t)\\ \bm{d}=\bm{d}^{\text{tr}}(t)\end{smallmatrix}} = \left[\begin{array}{c|c} \zeros{10}{10} & \eye{10} \\ \hline \multicolumn{2}{c}{\bar{\bm{A}}_\text{c} (\bm{\zeta}^\text{tr}(t),\dot{\bm{\zeta}}^\text{tr}(t),\bm{u}^\text{tr}(t))} \end{array}\right] \in \setrealmat{20}{20} \text{,} \label{eq:control_linearizationAc} \\
& \bm{B}_\text{c}(t) = \left.\frac{\partial \bm{\varphi}(\bm{x}, \bm{u}, \bm{d})}{\partial \bm{u}} \right|\!\!{\begin{smallmatrix}\bm{x}=\bm{x}^{\text{tr}}(t)\\ \bm{u}=\bm{u}^{\text{tr}}(t)\\ \bm{d}=\bm{d}^{\text{tr}}(t)\end{smallmatrix}} = \begin{bmatrix} \zeros{10}{4} \\ \bm{M}(\bm{\zeta}^\text{tr}(t))^{-1} \bm{L}_\text{in}(\bm{\zeta}^\text{tr}(t)) \end{bmatrix} \in \setrealmat{20}{4} \text{,} \\
& \bm{F}_\text{c}(t) = \left.\frac{\partial \bm{\varphi}(\bm{x}, \bm{u}, \bm{d})}{\partial \bm{d}} \right|\!\!{\begin{smallmatrix}\bm{x}=\bm{x}^{\text{tr}}(t)\\ \bm{u}=\bm{u}^{\text{tr}}(t) \\ \bm{d}=\bm{d}^{\text{tr}}(t)\end{smallmatrix}} = \begin{bmatrix} \zeros{10}{3} \\ \bm{M}(\bm{\zeta}^\text{tr}(t))^{-1} \bm{L}_\text{db} \end{bmatrix} \in \setrealmat{20}{3}  \text{.}
\end{align}

The linearized state-space equations \eqref{eq:estimation_statespacelinear} are then evaluated around an equilibrium point, resulting in a time-invariant system, and discretized using the zero-order-holder (ZOH) method for a given sampling time $T_s$, yielding\footnote{In order to avoid misleading, $(\cdot)^\text{tr} \triangleq (\cdot)^\text{eq}$ and $\delta (\cdot) \triangleq \Delta (\cdot)$.} 
\begin{equation} \label{eq:estimation_statespacelineardiscrete}
	\Delta \bm{x}_{k} = \bm{A}_\text{d} \Delta \bm{x}_{k-1} + \bm{B}_\text{d} \Delta \bm{u}_{k-1} + \bm{F}_\text{d} \bm{d}_{k-1} + \bm{w}_{k-1} \text{,} \\
\end{equation}
with `d' denoting discrete-time, $\bm{A}_\text{d} \in \setrealmat{20}{20}$, $\bm{B}_\text{d} \in \setrealmat{20}{4}$ and $\bm{F}_\text{d} \in \setrealmat{20}{3}$, with $\bm{w} \in \setrealvec{20}$ corresponding to unmodeled dynamics associated with linearization (truncated terms of the Taylor series expansion).

To improve state estimation, the external disturbances are also estimated by augmenting the state vector, as
\begin{equation} \label{eq:estimation_predictionmodelF}
	\underbrace{\begin{bmatrix} \Delta \bm{x}_k \\ \bm{d}_k \end{bmatrix}}_{\bm{\nu}_{k}} = \underbrace{\begin{bmatrix} \bm{A}_\text{d} & \bm{F}_\text{d} \\ \zeros{3}{20} & \eye{3} \end{bmatrix}}_{\Anu} \underbrace{\begin{bmatrix} \Delta \bm{x}_{k-1} \\ \bm{d}_{k-1} \end{bmatrix}}_{\bm{\nu}_{k-1}} + \underbrace{\begin{bmatrix} \bm{B}_\text{d} \\ \zeros{3}{4} \end{bmatrix}}_{\Bnu} \Delta \bm{u}_{k-1} + \underbrace{\begin{bmatrix} \bm{w}_{k-1} \\ \tilde{\bm{d}}_{k-1} \end{bmatrix}}_{\bar{\bm{w}}_{k-1}} \text{,}
\end{equation}
where $\tilde{\bm{d}}_{k-1} \triangleq \bm{d}_{k} - \bm{d}_{k-1}$. Moreover, linearizing the measurement equation \eqref{eq:estimation_measurementequations} around an equilibrium point yields
\begin{equation}
\bm{y}_{k} = \bm{\pi}(\bm{x}^\text{eq}) + \bm{H}_\text{d} \Delta \bm{x}_{k} + \bm{v}_k, \quad \bm{H}_\text{d} \triangleq \left. \frac{\partial \bm{\pi}(\bm{x})}{\partial \bm{x}} \right|_{\bm{x} = \bm{x}^\text{eq}} \in \setrealmat{16}{20} \text{,}
\end{equation}
with $\bm{v}_k$ now including unmodeled dynamics due to linearization. The last equation can be rewritten as
\begin{equation}  \label{eq:estimation_predictionmodelH}
\bm{y}_k = \underbrace{\begin{bmatrix} \bm{H}_\text{d} & \zeros{16}{3} \end{bmatrix}}_{\Hnu} \underbrace{\begin{bmatrix} \Delta \bm{x}_k \\ \bm{d}_k \end{bmatrix}}_{\bm{\nu}_k} + \underbrace{\bm{\pi}(\bm{x}^\text{eq}) + \bm{v}_k}_{\bar{\bm{v}}_k} \text{.}
\end{equation}

Now, taking into account the different sampling times, yields
\begin{equation}  \label{eq:estimation_switchedpredictionmodelH}
	\bm{y}_k = \underbrace{\begin{bmatrix} \bm{H}_\text{d}^{[k]} & \zeros{\iota_k}{3} \end{bmatrix}}_{\Hnu^{[k]}} \underbrace{\begin{bmatrix} \Delta \bm{x}_k \\ \bm{d}_k \end{bmatrix}}_{\bm{\nu}_k} + \underbrace{\bm{\pi}^{[k]}(\bm{x}^\text{eq}) + \bm{v}_k}_{\bar{\bm{v}}_k}, \quad \bm{H}^{[k]}_\text{d} \triangleq \left. \frac{\partial \bm{\pi}^{[k]}(\bm{x})}{\partial \bm{x}} \right|_{\bm{x} = \bm{x}^\text{eq}} \in \setrealmat{\iota_k}{20} \text{,}
\end{equation}
where $\bm{\pi}^{[k]}(\bm{x})$ is given by \eqref{eq:estimation_switchedmapping}.

\subsection{Zonotopic state estimator}
 
This subsection proposes the design of the zonotopic state estimator for the tilt-rotor UAV with suspended load, which extends the state estimation algorithm of \cite{Alamo2005a}, based on zonotopes and strips, assuming measurements with different sampling times.

Let $\setintval$ denote the set of real compact intervals. Define $\intval{a} = [ \lbound{a},\: \ubound{a} ] \triangleq \{a : \lbound{a} \leq a \leq \ubound{a}, \; \lbound{a},\: \ubound{a} \in \mathbb{R}\} \in \setintval$. Then, interval arithmetic operations are defined as $\intval{a} \odot \intval{b} \triangleq \{ a \odot b : a \in \intval{a},\: b \in \intval{b} \}$, with $\odot$ denoting `$+$', `$-$', `$\cdot$' or `$/$'. Elementary functions of $\intval{a}$, such as $\sin(\intval{a})$, are defined through their ranges over $\intval{a}$. Furthermore, $\midpoint{\intval{a}} \triangleq (1/2)(\lbound{a}+\ubound{a})$ and $\diam{\intval{a}} \triangleq (\ubound{a}-\lbound{a})$. Interval vectors and matrices are denoted by $\intval{\bm{a}}$ and $\intval{\bm{A}}$, respectively, for which $\midpoint{\cdot}$ and $\diam{\cdot}$ are defined component-wise. Moreover, an interval extension of a real valued function $f$ is denoted by $\square\{f\}$~\cite{Moore2009}.

Let the unitary interval be defined as $\ubox \triangleq [-1,\:1]$. Then $\ubox^r$ denotes a $r$-dimensional unitary box. An affine transformation of $\ubox^r$, given by $\{ \bm{c} + \bm{G}\bm{b} : \bm{b} \in \ubox^r \} = \bm{c} \oplus \bm{G}\ubox^r$, defines a $r$-order zonotope, with center $\bm{c} \in \setrealvec{n}$ and generator matrix $\bm{G} \in \setrealmat{n}{r}$, where $\oplus$ denotes the Minkowski sum of sets. A family of zonotopes, generated by an affine transformation of $\ubox^r$ through an interval matrix, is denoted by $\zonotopefamily = \bm{c} \oplus \intval{\bm{G}} \ubox^r \triangleq \{ \bm{c} + \bm{G} \bm{b}: \bm{G} \in \intval{\bm{G}}, \bm{b} \in \ubox^r \}$. A zonotope inclusion is defined by $\zinclusion{\zonotopefamily} \triangleq \bm{c} \oplus [ \midpoint{\intval{\bm{G}}} \; \bm{H}] \ubox^{n+r}$, with $\bm{H}_{ii} \triangleq (1/2) \sum_{j=1}^{r} \diam{\intval{\bm{G}}}_{ij}$, such that $\zonotopefamily \subseteq \zinclusion{\zonotopefamily}$. A strip is defined by $\strip \triangleq \{ \bm{x} \in \setrealvec{n} : |\bm{\rho}^T \bm{x} - \gamma| \leq \sigma \}$, with $\gamma$, $\sigma \in \setreal$, $\bm{\rho} \in \setrealvec{n}$ \cite{Alamo2005a}.

Consider the nonlinear discrete-time system 
\begin{equation} \label{eq:estimation_ZSEalgnonlinearsystem}
	\begin{aligned}
		\bm{x}_{k} & = \bm{f}(\bm{x}_{k-1}, \bm{w}_{k-1}) \text{,}\\
		\bm{y}_{k} & = \bm{g}(\bm{x}_{k}, \bm{v}_{k}) \text{,}
	\end{aligned}
\end{equation}
where $\bm{x}_{k} \in \setrealvec{n_x}$ are the state variables, $\bm{y}_{k} \in \setreal^{n_y}$ are the measured outputs, $\bm{w}_{k} \in \setrealvec{n_w}$ corresponds to process disturbances and parametric uncertainties, and $\bm{v}_{k} \in \setrealvec{n_v}$ represents measurement noise. Assume that $\bm{w}_k$, $\bm{v}_k$, and $\bm{x}_0$ belong to known compact sets $\mathbb{W}$, $\mathbb{V}$, and $\mathbb{X}_0$, respectively. %
Given the compact set $\mathbb{X}_{k-1}$, such that $\bm{x}_{k-1} \in \mathbb{X}_{k-1}$, the \emph{uncertain trajectory} of the system \eqref{eq:estimation_ZSEalgnonlinearsystem}, denoted by $\bm{f}(\mathbb{X}_{k-1}, \mathbb{W})$, is defined as the set of values that the time-update equation $\bm{f}$ achieves for all $\bm{x}_{k-1} \in \mathbb{X}_{k-1}$ and $\bm{w} \in \mathbb{W}$. Moreover, given the measured output $\bm{y}_k$, the \emph{consistent state set} is defined as $\mathbb{X}_{y_k} \triangleq \{ \bm{x} \in \setrealvec{n_x} : \bm{y}_k \in \bm{g}(\bm{x}, \mathbb{V}) \}$, and the \emph{exact uncertain state set} is defined by $\mathbb{X}_k \triangleq \bm{f}(\mathbb{X}_{k-1}. \mathbb{W}) \cap \mathbb{X}_{y_k}$. 
%

According to \cite{Kuhn1998}, uncertain trajectories of discrete-time systems can be bounded by zonotopes with sub-exponential overestimation. Consider the following theorems, presented in \citep{Alamo2005a}, whose proofs are omitted.

\begin{theorem}[Generalization of K\"{u}hn's method] \label{the:estimation_uncertaintrajectory} 
	Given a function $\bm{f}(\bm{x},\bm{w})$ with $\bm{x} \in \mathbb{X} \subset \setreal^{n_x}$ and $\bm{w} \in \mathbb{W} \subset \setreal^{n_w}$, in which $\mathbb{X} \triangleq \bm{c}_x \oplus \bm{G}_x \ubox^{r_x}$ and $\mathbb{W} \triangleq \bm{c}_w \oplus \bm{G}_w \ubox^{r_w}$ are known zonotopes. Define a zonotope $\zonotope_q \triangleq \bm{c}_q \oplus \bm{G}_q \ubox^{r_q}$ such that $\bm{f}(\bm{c}_x, \mathbb{W}) \subseteq \zonotope_q$; an interval matrix $\intval{\bm{M}} \triangleq \square\{\nabla_{\bm{x}} \bm{f}(\mathbb{X}, \mathbb{W})\} \bm{G}_x$; and a zonotope $\zonotope_\Psi \triangleq \zonotope_q \oplus \zinclusion{\intval{\bm{M}} \ubox^{r_x}}$. Then $\bm{f}(\mathbb{X}, \mathbb{W}) \subseteq \zonotope_\Psi$.
\end{theorem}

\begin{theorem} \label{the:estimation_consistentstates} 
	Given the zonotope $\bar{\mathbb{X}}_k \subset \setreal^{n_x}$ and the $i$-th measured output $\bm{y}_k(i)$. Define $\bm{\rho} \in \setreal^{n_x}$, $s \in \setreal$ and $\sigma \in \setreal$, obtained through interval arithmetic, such that $\bm{\rho} = \midpoint{\square\{\nabla_{\bm{x}} \bm{g}(i) (\bar{\mathbb{X}}_k, \mathbb{V})\}}$ and $\bm{\rho}^T \bar{\mathbb{X}}_k - \bm{g}(i)(\bar{\mathbb{X}}_k, \mathbb{V}) \subseteq [s - \sigma,\: s + \sigma]$. Then, $\bar{\mathbb{X}}_k \cap \mathbb{X}_{y_k(i)} \subseteq \bar{\mathbb{X}}_k \cap \bar{\mathbb{X}}_{y_k(i)}$, where $\bar{\mathbb{X}}_{y_k(i)} \triangleq \{ \bm{x} \in \setreal^{n_x} : |\bm{\rho}^T \bm{x} - \bm{y}_k(i) - s| \leq \sigma \}$.	
\end{theorem}

\begin{theorem} \label{the:estimation_intersection}
	Given a zonotope $\zonotope \triangleq \bm{c} \oplus \bm{G} \ubox^r \subset \setreal^{n}$, a strip $\strip \triangleq \{ \bm{x} \in \setreal^{n} : | \bm{\rho}^T \bm{x} - \gamma| \leq \sigma \}$ and a vector $\bm{\lambda} \in \setreal^{n}$. Define $\bm{c}_I(\bm{\lambda}) \triangleq \bm{c} + \bm{\lambda}(\gamma - \bm{\rho}^T \bm{c})$ and $\bm{G}_I(\bm{\lambda}) \triangleq [(\eye{n} - \bm{\lambda}\bm{\rho}^T)\bm{G} \;\, \sigma \bm{\lambda}]$. Then, $\zonotope \cap \strip \subseteq \zonotope_I(\bm{\lambda}) \triangleq \bm{c}_I(\bm{\lambda}) \oplus \bm{G}_I(\bm{\lambda}) \ubox^{r+1}$.
\end{theorem}

\begin{theorem} \label{the:estimation_minimumsegments} 
	Let $\zonotope_\text{I}(\bm{\lambda}) = \bm{c}_\text{I}(\bm{\lambda}) \oplus \bm{G}_\text{I}(\bm{\lambda}) \ubox^{r+1} \subset \setreal^{n}$, where $\bm{c}_\text{I}(\bm{\lambda}) \triangleq \bm{c} + \bm{\lambda}(\gamma - \bm{\rho}^T \bm{c})$ and $\bm{G}_\text{I}(\bm{\lambda}) \triangleq [(\eye{n} - \bm{\lambda}\bm{\rho}^T)\bm{G} \; \sigma \bm{\lambda}]$. Then, $\bm{\lambda} = (\bm{G} \bm{G}^T \bm{\rho})/(\bm{\rho}^T \bm{G} \bm{G}^T \bm{\rho} + \sigma^2)$ minimizes the Frobenius norm of $\bm{G}_\text{I}(\bm{\lambda})$.
\end{theorem}

Assume that a previously estimated set $\hat{\mathbb{X}}_{k-1}$ is available. Then, a zonotope $\bar{\mathbb{X}}_{k}$ bounding the uncertain trajectory $\bm{f}(\hat{\mathbb{X}}_{k-1}, \mathbb{W})$ can be obtained through Theorem \ref{the:estimation_uncertaintrajectory}. This operation is called \emph{prediction step} \cite{Le2013}, as an analogy to the Kalman filter algorithm. %
Moreover, given the $i$-th measured output $\bm{y}_k(i)$ and the zonotope $\bar{\mathbb{X}}_{k}$, a strip $\bar{\mathbb{X}}_{y_k(i)}$ can be computed through Theorem \ref{the:estimation_consistentstates}, such that $\bar{\mathbb{X}}_k \cap \mathbb{X}_{y_k(i)} \subseteq \bar{\mathbb{X}}_k \cap \bar{\mathbb{X}}_{y_k(i)}$. Then, a zonotope bounding $\bar{\mathbb{X}}_k \cap \bar{\mathbb{X}}_{y_k(i)}$ is obtained through Theorem \ref{the:estimation_intersection}. These operations together are called \emph{update step}. The resulting zonotope is parametrized by a vector $\bm{\lambda} \in \setrealvec{n_x}$, which is chosen according to specific criteria. A choice that minimizes the Frobenius norm of its generator matrix is given by Theorem \ref{the:estimation_minimumsegments}. Furthermore, Algorithm \ref{alg:estimation_Zorderredduction}, proposed by \cite{Combastel2003}, can be used to prevent the complexity of $\hat{\mathbb{X}}_k$ from increasing indefinitely, by computing a lower-order zonotope bounding $\hat{\mathbb{X}}_k$.

\begin{algorithm}
	\caption{Zonotope order reduction algorithm}
	\label{alg:estimation_Zorderredduction}
	\begin{algorithmic}[1]
		\Procedure{order\_reduction}{$\hat{\mathbb{X}}_{k}, r_\text{max}$}
		\State $\bm{H} \gets $ columns of $\bm{G}_{\hat{x}_k}$ ordered in decreasing Euclidean norm
		\State $\bm{H}_T \gets $ first $r_\text{max}-n_x$ columns of $\bm{H}$
		\For {$i = 1,...,n_x$}
		\State $\bm{Q}_{ii} \gets \sum_{j=r_\text{max}-n_x+1}^{r_{\hat{x}_k}} |\bm{H}_{ij}|$
		\EndFor
		\State $\hat{\mathbb{X}}_k \gets \bm{c}_{\hat{x}_k} \oplus [\bm{H}_T \; \bm{Q}] \ubox^{r_\text{max}}$ 
		\State \textbf{return} $\hat{\mathbb{X}}_k$
		\EndProcedure
	\end{algorithmic}
\end{algorithm}

The ZSE is summarized in Algorithm \ref{alg:estimation_ZSEalg}. It can be applied to multi-output systems by performing the update step using each element of the measured output vector iteratively \citep{Le2013}. Moreover, the possibility of dealing with measurements individually allows one to handle situations in which sensors have different sampling times in a straightforward manner.

\begin{algorithm}[htb]
	\caption{Zonotopic state estimator algorithm}
	\label{alg:estimation_ZSEalg}
	\begin{algorithmic}[1]
		\State Compute the zonotope $\bar{\mathbb{X}}_k \supseteq \bm{f}(\hat{\mathbb{X}}_{k-1},\mathbb{W})$ through Theorem \ref{the:estimation_uncertaintrajectory}
		\State Compute the strip $\bar{\mathbb{X}}_{y_k}$ through Theorem \ref{the:estimation_consistentstates}
		\State Compute the zonotope $\hat{\mathbb{X}}_k (\bm{\lambda}) \supseteq \bar{\mathbb{X}}_k \cap \bar{\mathbb{X}}_{y_k(i)}$ through Theorem \ref{the:estimation_intersection}
	\end{algorithmic}
\end{algorithm}

Although the zonotopic state estimator is formulated for nonlinear systems, the computational effort of computing the prediction step for \eqref{eq:modeling_statespacenonlinear} is very high\footnote{Theorem \ref{the:estimation_uncertaintrajectory} requires online computation of interval extensions over the Jacobian of the time-update equation. However, due to limited computational resources, an analytical expression for $\bm{M}(\bm{q})^{-1}$ could not be obtained, hence neither for $\bm{\varphi}(\bm{x},\bm{u},\bm{d})$, and consequently for the associated Jacobian.}. Therefore, the linearized system \eqref{eq:estimation_statespacelineardiscrete} is used instead. Moreover, the whole predicted set $\bar{\mathbb{X}}_k$ appears at least twice in the computation of the parameters $s$ and $\sigma$ (Theorem \ref{the:estimation_consistentstates}), which is performed through interval arithmetic. Therefore, due to interval dependency \cite{Moore2009}, it may yield a very large strip such that the intersection is $\bar{\mathbb{X}}_k$ itself. The next proposition shows that this problem can be avoided if the measurement equation is linear.
\begin{proposition}
	Consider the predicted zonotope $\bar{\mathbb{X}}_k$ and the linear measurement equation $\bm{y}_k = \bm{g}(\bm{x}_k,\bm{v}_k) \triangleq \bm{H} \bm{x}_k + \bm{v}_k$, where $\bm{y}_k \in \setrealvec{n_y}$ are the measured outputs, $\bm{x}_k \in \setrealvec{n_x}$ are the system states, $\bm{v}_k \in \mathbb{V}$, $\mathbb{V} \triangleq \bm{c}_v \oplus \bm{G}_v \ubox^{r_v} \subset \setrealvec{n_y}$, corresponds to measurement noise, and $\bm{H} \in \setrealmat{n_y}{n_x}$. If interval extensions are performed at a proper step, the strip obtained through Theorem \ref{the:estimation_consistentstates} using the $i$-th measurement $\bm{y}_k(i)$ is invariant with respect to $\bar{\mathbb{X}}_k$.
\end{proposition}
\begin{proof}
	Let $\bm{y}_k(i) = \bm{g}(i)(\bm{x}_k,\bm{v}_k)$, where $(i)$ denotes $i$-th line. From Theorem \ref{the:estimation_consistentstates},
	\begin{align*}
	\bm{\rho} & = \midpoint{\iextension{}{\nabla_{\bm{x}} (\bm{g}(i)(\bm{x}_k,\bm{v}_k))}|^{\bm{x}_k = \bar{\mathbb{X}}_k}_{\bm{v}_k = \mathbb{V}}} = \midpoint{\iextension{}{\nabla_{\bm{x}} (\bm{H}(i) \bm{x}_k + \bm{v}_k(i))}|^{\bm{x}_k = \bar{\mathbb{X}}_k}_{\bm{v}_k = \mathbb{V}}} \\
	& = \midpoint{\iextension{}{\bm{H}(i)^T}|^{\bm{x}_k = \bar{\mathbb{X}}_k}_{\bm{v}_k = \mathbb{V}}} = \midpoint{\bm{H}(i)^T} = \bm{H}(i)^T \text{.}
	\end{align*}
	
	Define $\bm{h}(\bm{x}_k,\bm{v}_k) \triangleq \bm{\rho}^T\bm{x}_k - \bm{g}(i)(\bm{x}_k,\bm{v}_k)$. We have that $\bm{h}(\bm{x}_k,\bm{v}_k) = \bm{\rho}^T\bm{x}_k - (\bm{H}(i) \bm{x}_k + \bm{v}_k(i)) = \bm{H}(i)\bm{x}_k - \bm{H}(i) \bm{x}_k - \bm{v}_k(i) = -\bm{v}_k(i)$. Then, $\iextension{}{\bm{h}(\bm{x}_k,\bm{v}_k)}\!|^{\bm{x}_k = \bar{\mathbb{X}}_k}_{\bm{v}_k = \mathbb{V}}$ gives
	\begin{align*}
	\iextension{}{-\bm{v}_k(i)}|^{\bm{x}_k = \bar{\mathbb{X}}_k}_{\bm{v}_k = \mathbb{V}} & = -\left(\bm{c}_v \oplus \text{rs}\left(\bm{G}_v\right) \ubox^{r_v} \right)(i) = - ( \bm{c}_v(i) \oplus \sum_{j=1}^{r_{v}} |\bm{G}_{v}(i,j)| \ubox ) \\
	& = [- \bm{c}_v(i) - \sum_{j=1}^{r_{v}} |\bm{G}_{v}(i,j)|, ~ - \bm{c}_v(i) + \sum_{j=1}^{r_{v}} |\bm{G}_{v}(i,j)| ] \triangleq [s - \sigma,~ s + \sigma] \text{,}
	\end{align*}
	where $(i,j)$ denotes $i$-th line and $j$-th column, and $\text{rs}(\cdot)$ denotes row sum \cite{Kuhn1998}. Hence, $\bm{\rho} = \bm{H}(i)^T$, $s = - \bm{c}_v(i)$, and $\sigma = \sum_{j=1}^{r_{v}} |\bm{G}_{v}(i,j)|$, which do not depend on $\bar{\mathbb{X}}_k$.
\end{proof}


\textred{We now propose the zonotopic state estimation algorithm for the tilt-rotor UAV with suspended load, considering the time-update equation \eqref{eq:estimation_predictionmodelF} and measurement \eqref{eq:estimation_predictionmodelH}.}

\begin{proposition}
Let $\bm{\nu}_{k-1}$, $\bar{\bm{w}}$ and $\bar{\bm{v}}$ belong to zonotopes $\hat{\mathbb{X}}_{k-1} \triangleq \bm{c}_{\hat{x}_{k-1}} \oplus \bm{G}_{\hat{x}_{k-1}} \ubox^{r_{\hat{x}_{k-1}}}$, $\bar{\mathbb{W}} \triangleq \bm{c}_{\bar{w}} \oplus \bm{G}_{\bar{w}} \ubox^{r_{\bar{w}}}$ and $\bar{\mathbb{V}} \triangleq \bm{c}_{\bar{v}} \oplus \bm{G}_{\bar{v}} \ubox^{r_{\bar{v}}}$, respectively. Denote $\mathbb{I}_k$ as the set of available measurements at time instant $k$, which is given according to the sensors' sampling times. Applying Algorithm \ref{alg:estimation_ZSEalg} for equations \eqref{eq:estimation_predictionmodelF} and \eqref{eq:estimation_predictionmodelH}, being the update step performed iteratively for each measurement available at time instant $k$, leads to Algorithm \ref{alg:estimation_ZSEtiltrotor}. 
\end{proposition}
\begin{proof}
	Step 1 of Algorithm \ref{alg:estimation_ZSEtiltrotor} is obtained from application of Theorem \ref{the:estimation_uncertaintrajectory} to \eqref{eq:estimation_predictionmodelF}, through the use of closed operations for linear image and Minkowski sum of zonotopes \cite{Kuhn1998}. Steps 3 to 11 are obtained from iterative use of Theorems \ref{the:estimation_consistentstates} through \ref{the:estimation_minimumsegments} for each measurement available based on \eqref{eq:estimation_predictionmodelH}, and step 12 correspond to the use of Algorithm \ref{alg:estimation_Zorderredduction} to limit the order of the estimated zonotope.
\end{proof}

\begin{algorithm}
	\caption{Zonotopic state estimator for the tilt-rotor UAV with suspended load}
	\label{alg:estimation_ZSEtiltrotor}
	\begin{algorithmic}[1]
		\Procedure{ZSE}{$\hat{\mathbb{X}}_{k-1},\Delta \bm{u}_{k-1}, \bm{y}_k, \bar{\mathbb{W}}, \bar{\mathbb{V}}, \mathbb{I}_k, r_\text{max}$}
		\State $\bar{\mathbb{X}}_k \gets (\bm{A}_{\bm{\nu}} \bm{c}_{\hat{x}_{k-1}} + \bm{B}_{\bm{\nu}} \Delta \bm{u}_{k-1} + \bm{c}_{\bar{w}}) \oplus [ \bm{G}_{\bar{w}} \;\; \bm{A}_{\bm{\nu}} \bm{G}_{\hat{x}_{k-1}}] \ubox^{r_{\bar{w}} + r_{\hat{x}_{k-1}}}$
		\State $\tilde{\mathbb{X}}_k \gets \bar{\mathbb{X}}_k$
		\ForAll {$ i \in \mathbb{I}_k$}
		\State $\bm{\rho} \gets \bm{H}_{\bm{\nu}}(i)^T$
		\State $s \gets - \bm{c}_{\bar{v}}(i)$
		\State $\sigma \gets \sum_{j=1}^{r_{\bar{v}}} |\bm{G}_{\bar{v}}(i,j)|$
		\State $\bm{\lambda} \gets (\bm{G}_{\tilde{x}_k} \bm{G}_{\tilde{x}_k}^T \bm{\rho})/(\bm{\rho}^T \bm{G}_{\tilde{x}_k} \bm{G}_{\tilde{x}_k}^T \bm{\rho} + \sigma^2)$
		\State $\tilde{\mathbb{X}}_k \gets (\bm{c}_{\tilde{x}_k} + \bm{\lambda}(\bm{y}_k(i) + s - \bm{\rho}^T \bm{c}_{\tilde{x}_k})) \oplus [(\eye{23} - \bm{\lambda}\bm{\rho}^T)\bm{G}_{\tilde{x}_k} \; \sigma \bm{\lambda}] \ubox^{r_{\tilde{x}_k}+1}$ 
		\EndFor
		\State $\hat{\mathbb{X}}_k \gets \tilde{\mathbb{X}}_k$
		\State $\hat{\mathbb{X}}_k \gets \text{order\_reduction}(\hat{\mathbb{X}}_k,r_\text{max})$
		\State \textbf{return} $\hat{\mathbb{X}}_k$
		\EndProcedure
	\end{algorithmic}
\end{algorithm}

The design of the zonotopic state estimator lies in the appropriate choice of the zonotopes $\bar{\mathbb{W}}$ and $\bar{\mathbb{V}}$. On the other hand, to obtain $\hat{\mathbb{X}}_0$ a zonotope $\bar{\mathbb{X}}_0$ containing the system's initial states must be known, to which an \emph{initial update step} is applied using initial measurements $\bm{y}_0$. As a drawback from using a linearized model for the zonotopic state estimation algorithm, the property $\bm{\nu}_k \in \hat{\mathbb{X}}_k$ is guaranteed only if the chosen zonotopes $\bar{\mathbb{W}}$ and $\bar{\mathbb{V}}$ contain all the unmodelled dynamics due to linearization. Moreover, by augmenting the state vector with the external disturbances, bounds must be assumed on their variations within the controller sampling time, instead of bounds on their magnitudes. The latter may result in reduced overestimation in cases with non-abrupt disturbances.

\subsection{\textred{Kalman filter}} \label{sec:estimation_LKF}


\textred{This section presents the derivation of a Kalman filter, for comparison purposes with the proposed zonotopic state estimation strategy.} To design a KF for the tilt-rotor UAV with suspended load based on \eqref{eq:estimation_predictionmodelF} and \eqref{eq:estimation_switchedpredictionmodelH}, $\bar{\bm{w}}$ is regarded as process noise. Moreover, $\bar{\bm{w}}$ and ${\bm{v}}$, are assumed to be white, mutually uncorrelated, zero-mean Gaussian distributions, with known constant covariance matrices denoted by $\bm{P}^{\bm{w}} \in \setrealmat{23}{23}$ and $\bm{P}^{\bm{v}[k]} \in \setrealmat{\iota_k}{\iota_k}$, respectively, where $\bm{P}^{\bm{v}[k]}$ is a diagonal matrix formed by all $\bm{P}^{\bm{v}}(i,i)$ such that $i \in \setI_k$, with $\bm{P}^{\bm{v}} \in \setrealmat{16}{16}$, and $(i,i)$ denoting $i$-th line and $i$-th column. %
Let $\hat{(\cdot)}$ denote estimated \textred{variables}, and $(\cdot)_{m|n}$ denote information at time instant $m$ given measurements up to instant $n$. Then, given a previous estimation $\hat{\bm{\nu}}_{k-1|k-1}$, \textred{the state vector $\hat{\bm{\nu}}_{k|k-1}$ is given by the \emph{prediction step} \cite{Simon2006}}
\begin{equation}
\hat{\bm{\nu}}_{k|k-1} = \bm{A}_{\bm{\nu}} \hat{\bm{\nu}}_{k-1|k-1} + \bm{B}_{\bm{\nu}} \Delta \bm{u}_{k-1} \text{,}
\end{equation}
whilst $\hat{\bm{\nu}}_{k|k}$ \textred{is given by the \emph{correction step}}, defined as
\begin{gather} \label{eq:estimation_LKFestimation}
\hat{\bm{\nu}}_{k|k} = \hat{\bm{\nu}}_{k|k-1} + \bm{N}_k (\bm{y}_k - (\Hnu^{[k]} \hat{\bm{\nu}}_{k|k-1} + \bm{\pi}^{[k]}(\bm{x}^\text{eq}))) \text{,}
\end{gather}
being $\bm{N}_k$ the \emph{Kalman gain}. Define $\tilde{\bm{\nu}}_k \triangleq \bm{\nu}_k - \hat{\bm{\nu}}_{k|k}$ as the estimation error, and let $\bm{P}^{\bm{\nu}}_{k|k} \triangleq \expect{\tilde{\bm{\nu}}_k\tilde{\bm{\nu}}_k^T}$, where $\expect{\cdot}$ denotes expectation. Then, the covariance propagations are computed from \eqref{eq:estimation_predictionmodelF} as \textred{$\bm{P}^{\bm{\nu}}_{k|k-1} = \bm{A}_{\bm{\nu}} \bm{P}^{\bm{\nu}}_{k-1|k-1} \bm{A}_{\bm{\nu}}^T + \bm{P}^{\bm{w}}$, and from the correction step as $\bm{P}^{\bm{\nu}}_{k|k} = (\eye{20} - \bm{N}_k \bm{H}^{[k]}_{\bm{\nu}}) \bm{P}^{\bm{\nu}}_{k|k-1} (\eye{20} - \bm{N}_k \bm{H}^{[k]}_{\bm{\nu}})^T + \bm{N}_k \bm{P}^{{\bm{v}}{[k]}} \bm{N}_k^T$}. %
The Kalman gain is \textred{obtained} such that the KF provides minimum variance estimation. \textred{The well-known solution of this optimization problem is the gain update equation}
\begin{equation}
\bm{N}_k = \bm{P}^{\bm{\nu}}_{k|k-1} (\bm{H}_{\bm{\nu}}^{[k]})^T (\bm{H}_{\bm{\nu}}^{[k]} \bm{P}^{\bm{\nu}}_{k|k-1} (\bm{H}_{\bm{\nu}}^{[k]})^T + \bm{P}^{\bm{v}{[k]}})^{-1} \text{.}
\end{equation}
\section{Control design} \label{sec:control}

This section presents the design of a state-feedback control strategy for trajectory tracking of the suspended load with stabilization of the tilt-rotor UAV. Since the load's position and orientation are represented by state variables, state-feedback strategies can directly steer the trajectory of the suspended load with respect to the inertial reference frame. Besides, the aircraft's behavior is implicit in the state-space equations, then stabilization of the system implies stabilization of the tilt-rotor UAV. \textred{Moreover, although the multi-body structure of the tilt-rotor UAV with suspended load system (see Section \ref{sec:modeling}) may not impose additional difficulties to the state estimation in comparison to other UAV configurations for load transportation, it generates challenges to the control design due to the dynamic couplings between the rigid bodies and the indirect input mapping from the tilting mechanism's torques to the load's pose \cite{Raffo2018}. Therefore, to cope with the resulting issues, the proposed control strategies are based on the whole-body dynamic equations developed in Section \ref{sec:modeling}, without requiring the formulation of cascade control structures.}

Since in future works the control algorithm will be implemented on an embedded system in the real aircraft, a discrete-time control strategy is designed. Therefore, %
based on discrete-time, augmented linearized error dynamics, a mixed $\mathcal{H}_2/\mathcal{H}_\infty$ controller with pole placement constraints is proposed. The controller features constant disturbances rejection, and achieves improved disturbance attenuation by mininizing the $\mathcal{H}_2$ norm of the closed-loop system while guaranteeing a specified upper-bound for its $\mathcal{H}_\infty$ norm. Furthermore, time response requirements are satisfied by imposing constraints in the pole placement process.

\subsection{Linearized parameter-varying error dynamics}  \label{sec:control_linearizedmodel}

The tilt-rotor UAV with suspended load is a mechanical system with more degrees of freedom than control inputs, therefore being characterized as an underactuated mechanical system. As it has four control inputs, only up to four degrees of freedom can be steered along a desired, arbitrary trajectory, while the remaining DOF can only be stabilized. Aiming path tracking control of the suspended load, the position $\bm{\xi} = [x \; y \; z]^T$ and yaw angle $\psi$ of the load are chosen to be regulated.

Define the auxiliary variable $\bm{\zeta}^\text{tr}(t) \triangleq [ \phi^\text{eq} \;\, \theta^\text{eq} \;\, \psi^\text{tr}(t) \;\, (\bm{\gamma}^\text{eq})^T \;\, \alpha_\text{R}^\text{eq} \;\, \alpha_\text{L}^\text{eq}]^T$. Assuming that $\psi^\text{tr}$ is constant for the desired trajectory, we have that $\bm{\zeta}^\text{tr}$ is also constant. Then, $\dot{\bm{\zeta}}^\text{tr} = \zeros{7}{1}$, and by defining $\bm{q}^\text{tr} \triangleq [ (\bm{\xi}^\text{tr})^T \;\, (\bm{\zeta}^\text{tr})^T ]^T$, we have that $\dot{\bm{q}}^\text{tr} = [ (\dot{\bm{\xi}}^\text{tr})^T \;\, \zeros{1}{7} ]^T$ and $\ddot{\bm{q}}^\text{tr} = [ (\ddot{\bm{\xi}}^\text{tr})^T \;\, \zeros{1}{7} ]^T$. Moreover, define $\bm{x}^\text{tr} \triangleq [(\bm{q}^\text{tr})^T \;\, (\dot{\bm{q}}^\text{tr})^T]^T$ and $\bm{u}^\text{tr} \triangleq [\infR^\text{tr} \;\, \infL^\text{tr} \;\, \intauaR^\text{tr} \;\, \intauaL^\text{tr} ]^T$. Evaluating \eqref{eq:estimation_statespacelinear} around $\bm{x}^\text{tr}$ and $\bm{u}^\text{tr}$, leads to the error dynamics $\delta \dot{\bm{x}} = \bm{A}_\text{c}(t) \delta \bm{x} + \bm{B}_\text{c} \delta \bm{u} + \bm{F}_\text{c} \bm{d}$, in which $\bm{B}_\text{c}$ and $\bm{F}_\text{c}$ are constant matrices, and $\bm{A}_\text{c}(t)$ is time-varying only due to $\bm{u}^\text{tr}(t)$ (see equation \eqref{eq:control_linearizationAc}).

For control design, this work assumes that the desired accelerations are not negligible, i.e., the time-varying matrix $\bm{A}_\text{c}(t)$ can not be approximated by $\bm{A}_\text{c}(t)|_{\bm{u}^\text{tr} = \bm{u}^\text{eq}}$. To take into account this fact in the control design approach, the desired accelerations are regarded as bounded uncertain parameters. Recalling the dynamics \eqref{eq:estimation_feasibletrajectory}, approximated values for $\bm{u}^\text{tr}(t)$ can be obtained by
\begin{equation} \label{eq:control_feedforwardcontinuous}
	\bm{u}^\text{tr}(t) = \bm{L}_\text{in}(\bm{q}^\text{tr})^+ \left[ \bm{M}(\bm{q}^\text{tr}) \ddot{\bm{q}}^\text{tr} + (\bm{C}(\bm{q}^\text{tr}, \dot{\bm{q}}^\text{tr}) + \bm{L_\text{fr}}) \dot{\bm{q}}^\text{tr} + \bm{g}(\bm{q}^\text{tr}) \right],
\end{equation}
where $\bm{L}_\text{in}(\bm{q}^\text{tr})^+$ denotes the left Moore-Penrose pseudo-inverse of $\bm{L}_\text{in}(\bm{q}^\text{tr})$. In the same lines of the discussion in Section \ref{sec:estimation_linearizedmodel}, it can be shown that \eqref{eq:control_feedforwardcontinuous} is not a function of $\bm{\xi}^\text{tr}$ and $\dot{\bm{\xi}}^\text{tr}$. Thus, by defining the vector of parameters $\bm{\sigma} \triangleq \ddot{\bm{\xi}}^\text{tr}$, we have that $\bm{u}^\text{tr}(t) \triangleq \bm{u}(\bm{\sigma})$, which is an affine function of $\bm{\sigma}$. Moreover, since the state-space equations \eqref{eq:modeling_statespacenonlinear} are affine in the inputs, the Jacobian \eqref{eq:control_linearizationAc} is also affine. Thus, by substituting \eqref{eq:control_feedforwardcontinuous} in \eqref{eq:control_linearizationAc}, yields the parameter-varying matrix $\bm{A}_\text{c}(\bm{\sigma})$, which is also affine in the parameters $\bm{\sigma}$.

Furthermore, to improve the trajectory tracking and provide rejection to constant disturbances, the state vector $\delta\bm{x}$ is augmented with integral actions, computed by integrating the error of the regulated degrees of freedom, yielding
\begin{equation} \label{eq:control_augmentedstatevector}
	\bm{\chi} \triangleq \begin{bmatrix} \delta \bm{x} \\ \int (\bm{\xi}-\bm{\xi}^\text{tr}) \\ \int (\psi - \psi^\text{tr}) \end{bmatrix} \in \setrealvec{24} \text{,}
\end{equation}
whose parameter-varying dynamics are given by
\begin{equation} \label{eq:control_augmentedcontinuousdynamics}
	\dot{\bm{\chi}} = \underbrace{\left[\begin{array}{cccccc|c|c}
	\multicolumn{7}{c|}{\bm{A}_\text{c}(\bm{\sigma})}                & \zeros{20}{4}                 \\ \hline 
	1 & 0 & 0 & 0 & 0 & 0 & \multirow{4}{*}{\zeros{4}{14}} & \multirow{4}{*}{\zeros{4}{4}} \\
	0 & 1 & 0 & 0 & 0 & 0 &                                & 	                           \\                      
	0 & 0 & 1 & 0 & 0 & 0 &                                &                               \\
	0 & 0 & 0 & 0 & 0 & 1 &                                & 
	\end{array}\right]}_{\tilde{\bm{A}}_\text{c}(\bm{\sigma})} \bm{\chi} + \underbrace{\begin{bmatrix} \bm{B}_\text{c} \\ \zeros{4}{4} \end{bmatrix}}_{\tilde{\bm{B}}_\text{c}} \delta\bm{u} + \underbrace{\begin{bmatrix} \bm{F}_\text{c} \\ \zeros{4}{3} \end{bmatrix}}_{\tilde{\bm{F}}_\text{c}} \bm{d} \text{.}
\end{equation}


Finally, since the control design will be performed in discrete-time, the linear parameter-varying (LPV) system \eqref{eq:control_augmentedcontinuousdynamics} is discretized through Euler approximation for the sampling time $T_s$, yielding the discrete-time, augmented LPV error dynamics
\begin{equation} \label{eq:control_discretizedaugmentederrordynamics}
\bm{\chi}_{k+1} = \Achi (\bm{\sigma}) \bm{\chi}_{k} + \Bchi \delta \bm{u}_{k} + \Fchi \bm{d}_{k} \text{,}
\end{equation}
with $\Achi(\bm{\sigma}) \triangleq \eye{24} + T_s \tilde{\bm{A}}_\text{c}(\bm{\sigma}) \in \setrealmat{24}{24}$, $\Bchi \triangleq T_s \tilde{\bm{B}}_\text{c} \in \setrealmat{24}{4}$ and $\Fchi \triangleq T_s \tilde{\bm{F}}_\text{c} \in \setrealmat{24}{3}$. Since $\tilde{\bm{A}}_\text{c}(\bm{\sigma})$ is affine in $\bm{\sigma}$, note that $\Achi(\bm{\sigma})$ is also affine in the parameters $\bm{\sigma}$. Thus, assuming bounded desired accelerations, we have that the LPV system \eqref{eq:control_discretizedaugmentederrordynamics} can be rewritten in a convex polytopic representation, by defining $\Achi(\bm{\sigma}) \triangleq \Achi(\tilde{\bm{\sigma}}) = \sum_{i=1}^{8} \tilde{\sigma}_i\Achi^i$, with $\Achi^i$ denoting the $i$-th vertex of $\Achi(\tilde{\bm{\sigma}})$, and $\sum_{i=1}^{8} \tilde{\sigma}_i = 1$. Resulting errors from linearization and discretization will be taken into account as unmodeled dynamics, and the controller will be assumed to be robust enough to deal with the subsequent effects.

\subsection{Discrete-time mixed $\mathcal{H}_2/\mathcal{H}_\infty$ control}

The present mixed $\mathcal{H}_2/\mathcal{H}_\infty$ control paradigm is an extension for discrete-time linear systems of the method proposed in \cite{Gahinet1996}. In order to design a discrete-time mixed $\mathcal{H}_2/\mathcal{H}_\infty$ controller for trajectory tracking of the suspended load with stabilization of the tilt-rotor UAV, consider the discrete-time uncertain linear system 
\begin{equation}
\begin{aligned}
\bm{\chi}_{k+1} & = \Achi(\bm{\sigma}) \bm{\chi}_{k} + \Bchi \delta \bm{u}_{k} + \Fchi \bm{d}_{k}\text{,}\\
\bm{z}_k^{(2)} & = \bm{H}_{\bm{z}} \bm{\chi}_k + \bm{D}_{\bm{zu}} \delta\bm{u}_k \text{,} \\
\bm{z}_k^{(\infty)} & = \bm{H}_{\bm{z}} \bm{\chi}_k + \bm{D}_{\bm{zu}} \delta\bm{u}_k + \bm{D}_{\bm{zd}} \bm{d}_k \text{,}
\end{aligned}
\end{equation}
where $\bm{z}_k^{(2)}, \bm{z}_k^{(\infty)} \in \setreal^{n_z}$ are cost variables, $\bm{H}_{\bm{z}} \in \setrealmat{n_z}{24}$, $\bm{D}_{\bm{zu}} \in \setrealmat{n_z}{4}$, and $\bm{D}_{\bm{zd}} \in \setrealmat{n_z}{3}$ are weighting matrices. Let $\bm{\Psi}^{(2)}_{\bm{d}{\bm{z}}}(\zdomain)$ and $\bm{\Psi}^{(\infty)}_{\bm{d}{\bm{z}}}(\zdomain)$ denote the discrete-time transfer matrices from $\bm{d}_k$ to $\bm{z}_k^{(2)}$ and $\bm{d}_k$ to $\bm{z}_k^{(\infty)}$, respectively, \textred{and $\| \bm{\Psi}^{(2)}_{{\bm{d}}\bm{z}} \|_2$ and $\|\bm{\Psi}^{(\infty)}_{{\bm{d}}{\bm{z}}} \|_\infty$ denote the corresponding $\mathcal{H}_2$ and $\mathcal{H}_\infty$ norms.
The objective is to design a state-feedback controller of the form $\delta \bm{u}_k = - \bm{K} \bm{\chi}$ that minimizes $\| \bm{\Psi}^{(2)}_{{\bm{d}}\bm{z}} \|_2^2 \triangleq \sum_{k=0}^{\infty} \trace {\bm{\Psi}^{(2)}_{\bm{d}\bm{z},k}(\bm{\Psi}^{(2)}_{\bm{d}\bm{z},k})^T}$ while guaranteeing a specified upper-bound for $\| \bm{\Psi}^{(\infty)}_{{\bm{d}}{\bm{z}}} \|_\infty$, where $\bm{\Psi}^{(2)}_{\bm{d}\bm{z},k} \triangleq \mathcal{Z}^{-1}\{\bm{\Psi}^{(2)}_{\bm{d}\bm{z}}(\zdomain)\}$, with $\mathcal{Z}$ denoting the z-transform, ensuring better transient response and disturbance attenuation for the closed-loop system. The gain matrix $\bm{K}$ that minimizes $\trace{\bm{\Omega}} > \| \bm{\Psi}^{(2)}_{{\bm{d}}\bm{z}} \|_2^2$ while guaranteeing a given bound $\tilde{\gamma} > \| \bm{\Psi}^{(\infty)}_{{\bm{d}}{\bm{z}}} \|_\infty^2$, is computed by} $\bm{K} = -\bm{Y}\bm{X}^{-1}$, where $\bm{Y}$ and $\bm{X}$ are obtained by solving the optimization problem \cite{Oliveira2002}
%
\begin{gather}
\underset{\bm{P},\bm{X},\bm{Y},\bm{\Omega}}{\min} ~ \text{trace}\{\bm{\Omega}\} \quad \text{subject to} \nonumber\\
\begin{bmatrix}
\bm{\Omega} & \bm{H}_{\bm{z}} \bm{X} + \bm{D}_{\bm{zu}} \bm{Y} \\
* & \bm{X} + \bm{X}^T - \bm{P}
\end{bmatrix} > 0 \text{,} \label{eq:lmiH2first}\\
\begin{bmatrix}
\bm{P} & \Achi^i \bm{X} + \bm{B}_{\bm{\chi}} \bm{Y} & \bm{F}_{\bm{\chi}} \\
* & \bm{X} + \bm{X}^T - \bm{P} & \zeros{24}{3} \\
* & * & \eye{3} 
\end{bmatrix} > 0 \text{,} \label{eq:lmiH2second} \\
\begin{bmatrix}
\bm{P} & \Achi^i \bm{X} + \bm{B}_{\bm{\chi}} \bm{Y} & \bm{F}_{\bm{\chi}} & \zeros{24}{n_z} \\
* & \bm{X} + \bm{X}^T - \bm{P} & \zeros{24}{3} & \bm{X}^T \bm{H}_{\bm{z}}^T + \bm{Y}^T \bm{D}_{\bm{zu}}^T \\
* & * & \eye{3} & \bm{D}_{\bm{zd}}^T \\
* & * & * & \tilde{\gamma} \eye{n_z} 
\end{bmatrix} > 0 \text{,} \label{eq:lmiHinf}
\end{gather}
with $i = 1,2,\dots,8$, $\bm{P} = \bm{P}^T > 0$, $\bm{X} > 0$ and $\bm{\Omega} = \bm{\Omega}^T$. 


In order to guarantee time response specifications for the closed-loop system, constraints in the form of Linear Matrix Inequality (LMI) regions are imposed on the pole placement performed by the controller. An LMI region is defined as a convex subset of the complex plane that can be expressed as \textred{$\mathbb{D} \triangleq \{ \zdomain \in \mathbb{C} : \bm{U} + \zdomain \bm{V} + \zdomain^* \bm{V}^T < 0\}$ \citep{Gahinet1996},} 
whose shape is defined by the matrices $\bm{U} = \bm{U}^T \in \setrealmat{n_\text{D}}{n_\text{D}}$ and $\bm{V} \in \setrealmat{n_\text{D}}{n_\text{D}}$. Such regions are symmetric with respect to the real axis, and the intersection between two of them is also an LMI region.

To ensure minimum and maximum settling times, maximum percentage overshoot, and also to avoid the ringing effect, three regions are of interest (see Figure \ref{fig:control_LMIregions}): $\mathbb{D}_1 \triangleq \{ \zdomain \in \mathbb{C} : \real{\zdomain} > \deps \geq 0\}$, $\mathbb{D}_2 \triangleq \{ \zdomain \in \mathbb{C} : 0 \leq |\zdomain| < \varpi \}$ and $\mathbb{D}_3 \triangleq \{ \zdomain \in \mathbb{C} : 0 \leq |\imag{\zdomain}| < \tau \}$. The eigenvalues of the closed-loop system matrix $\Achi(\bm{\sigma}) - \Bchi \bm{K}$ belong to $\bigcap_{j=1}^{3} \mathbb{D}_j$ if, and only if, there exists a symmetric matrix $\bm{T} > 0$ such that
\begin{gather}
	\bm{T} (\Achi^i)^T + \Achi^i \bm{T} + \bm{Y}^T\bm{B}_{\bm{\chi}}^T + \bm{B}_{\bm{\chi}}\bm{Y} - 2\varepsilon \bm{T} > 0 \text{,} \label{eq:lmiDstability1} \\
	\begin{bmatrix}
	-\varpi \bm{T} & \Achi^i \bm{T} + \bm{B}_{\bm{\chi}} \bm{Y} \\
	* & -\varpi \bm{T}
	\end{bmatrix} < 0 \text{,} \label{eq:lmiDstability2} \\
	\begin{bmatrix} -2 \tau \bm{T} ~& ~\bm{T} (\Achi^i)^T \!{-} \Achi^i \bm{T} + \bm{Y}^T \Bchi^T {-} \Bchi \bm{Y} \\ * & - 2 \tau \bm{T} \end{bmatrix} < 0 \text{,} \label{eq:lmiDstability3}
\end{gather}
from which $\bm{K} = -\bm{Y} \bm{T}^{-1}$. The LMI constraints \eqref{eq:lmiDstability1}--\eqref{eq:lmiDstability3} are considered in the control design along with \eqref{eq:lmiH2first}--\eqref{eq:lmiHinf} by imposing $\bm{X} = \bm{X}^T = \bm{T} > 0$.

\begin{figure}[!tb]
	\centering{
		\def\svgwidth{0.8\columnwidth}
		{\footnotesize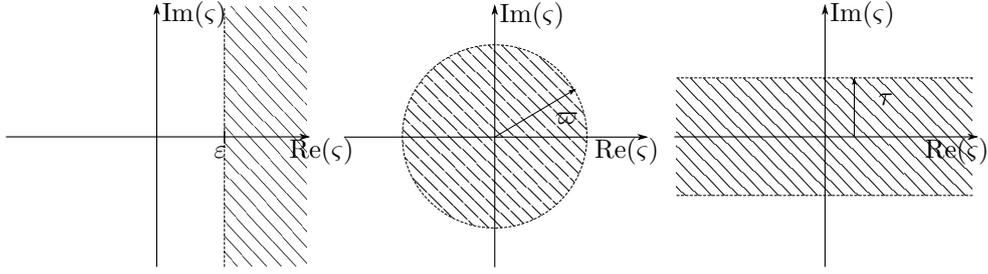}
		\caption{LMI regions considered in this work.}\label{fig:control_LMIregions}}
\end{figure}
%


Since the control signals provided by the mixed $\mathcal{H}_2/\mathcal{H}_\infty$ controller are associated with the linearized dynamics \eqref{eq:control_discretizedaugmentederrordynamics}, to apply the control signals to the tilt-rotor UAV with suspended load, \textred{the feed-forward term \eqref{eq:control_feedforwardcontinuous} is computed at each $k$. 
 The control law is then given by
\begin{equation} \label{eq:control_law}
	\bm{u}_k = \delta \bm{u}_k +  \bm{u}^\text{tr}_k = - \bm{K} \bm{\chi}_k + \bm{L}_\text{in}(\bm{q}^\text{tr}_k)^+ \left[ \bm{M}(\bm{q}^\text{tr}_k) \ddot{\bm{q}}^\text{tr}_k + (\bm{C}(\bm{q}^\text{tr}_k, \dot{\bm{q}}^\text{tr}_k) + \bm{L_\text{fr}}) \dot{\bm{q}}^\text{tr}_k + \bm{g}(\bm{q}^\text{tr}_k) \right] \text{,}
\end{equation}
where $\bm{q}^\text{tr}_k$, $\dot{\bm{q}}^\text{tr}_k$ and $\ddot{\bm{q}}^\text{tr}_k$ are reference signals given at instant $k$.}
%
%
On the other hand, since $\bm{u}^\text{tr}_k$ is a least-squares solution to \eqref{eq:control_feedforwardcontinuous}, which assumes a scenario without disturbances, it is an exact solution to the dynamic equations \eqref{eq:eulerlagrangeCANONICAL} only if the desired trajectory satisfies \eqref{eq:estimation_feasibletrajectory}. Moreover, the control signal $\bm{u}^\text{tr}_k$ will be sustained for $T_s$ seconds. The subsequent errors are also considered as unmodeled dynamics, and the controller is assumed to be robust enough to deal with these effects.

The described control strategy rely on full information about the system states \eqref{eq:modeling_systemstates} in order to achieve path tracking of the suspended load. Considering the scenario described in Section \ref{sec:estimation}, the feedback connection in \eqref{eq:control_law} is performed using an estimated state vector $\hat{\bm{\chi}}_k$, defined according to
\begin{equation}
	\hat{\bm{\chi}} \triangleq \begin{bmatrix} \hat{\bm{x}} - \bm{x}^\text{tr} \\ \int (\hat{\bm{\xi}}-\bm{\xi}^\text{tr}) \\ \int (\hat{\psi} - \psi^\text{tr}) \end{bmatrix} \in \setrealvec{24} \text{,}
\end{equation}
\textred{where the estimated states $\hat{\bm{x}}_k$ are obtained from the center of the zonotope provided by Algorithm \ref{alg:estimation_ZSEtiltrotor}. Equilibrium values are added and subtracted to the estimated states and control signals, to adapt them for the control strategy and state estimation algorithm, respectively, as shown in Figure \ref{fig:control_blockdiagram}. In the case of the Kalman filter, the estimated state vector is obtained from \eqref{eq:estimation_LKFestimation}.}

\begin{figure}[!htb]
	\centering{
		\def\svgwidth{0.8\columnwidth}
		{\scriptsize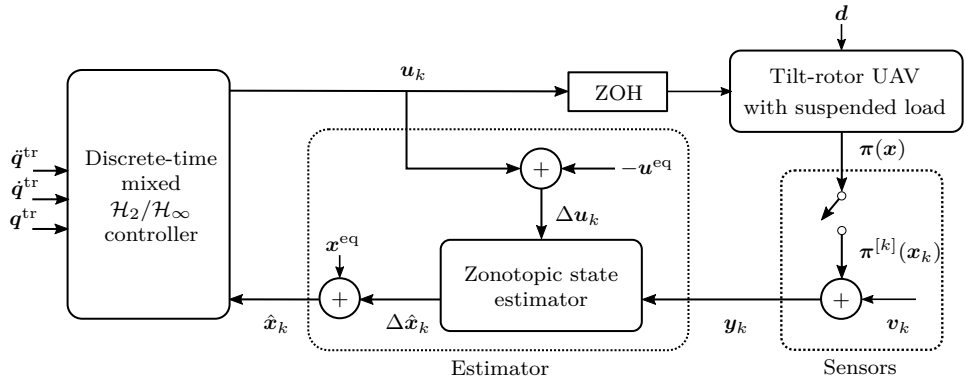}
		\caption{Control structure.}\label{fig:control_blockdiagram}}
\end{figure}

\section{Numerical experiments} \label{sec:results}

This section evaluates the performance of the proposed control and state estimation strategies through experiments in the ProVANT simulation environment.

\subsection{Experiment description}

The ProVANT simulator\footnote{\url{https://github.com/Guiraffo/ProVANT-Simulator}} is a simulation environment for tilt-rotor UAVs, developed in the ProVANT project\footnote{\textred{The ProVANT project is a collaborative work involving the brazilian universities Federal University of Santa Catarina and Federal University of Minas Gerais, and the University of Sevilla, Spain. The objective is the development and research of convertible UAVs in the tilt-rotor configuration.}}, based on the robotic applications framework Robot Operating System (ROS) \cite{Quigley2009} and the open-source robot simulation environment Gazebo \cite{Koenig2004}. Based on Computer Aided Design (CAD) 3D models, the main purpose of the ProVANT simulator is the validation of control strategies designed for tilt-rotor UAVs, in a stage previous to experiments in the real aircraft. %
Figure \ref{fig:results_cadmodel} illustrates the CAD model of the tilt-rotor UAV with suspended load, as shown in the simulation environment.

\begin{figure}[!htbp]
	\centering{
		\includegraphics[width = 0.3\textwidth]{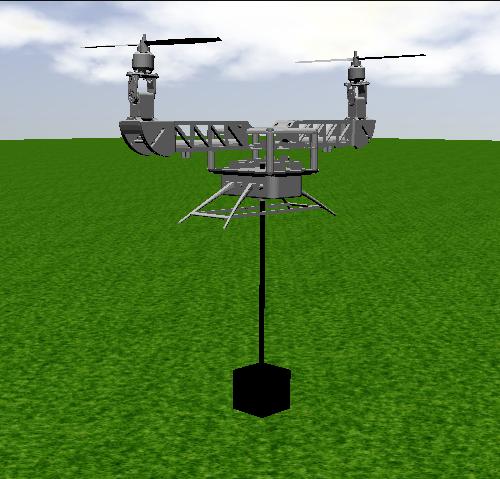}
		\caption{The tilt-rotor UAV with suspended load (CAD model), as shown in the ProVANT simulator.}\label{fig:results_cadmodel}}
\end{figure}

The experiment consists in performing trajectory tracking of the suspended load, with stabilization of the tilt-rotor UAV. The desired trajectory is composed of several connected paths, defined in Table \ref{tab:results_desiredtrajectory}. \textred{The initial position of the UAV is given by $x_\frB = 0$ m, $y_\frB = 0$ m and $z_\frB = 1.619$ m, the rope initial angles are given by $\gone = \gtwo = 15^\text{o}$, and the UAV initial orientation $\etaB$ and tilting angles $\aR$, $\aL$ are equal to zero.} This trajectory is proposed to evaluate the performance of the designed strategies in a scenario starting with vertical take-off in a spiral path, straight line following with rapid changing in direction, and vertical landing, with \textred{$\psi^\text{tr} = 0^\text{o}$.} Moreover, to evaluate the robustness against external disturbances of the proposed strategy, \textred{aerodynamic disturbance forces are applied to the suspended load during the experiment, resulting from environmental wind and drag. The disturbance forces are generated according to \cite{Oktay2013}
\begin{equation}
	\bm{d}_i^\frL = \frac{1}{2} \rho_\text{air} a_\text{S} |\bm{\kappa}_i^\frL| \bm{\kappa}_i^\frL,
\end{equation}
where $\bm{d}_i^\frL$ denotes the $i$-th component of $\bm{d}^\frL$, from which the disturbance force vector in $\frI$ is computed by $\bm{d} = \RIL \bm{d}^\frL$, with $\rho_\text{air}$ the air density, $a_\text{S}$ the equivalent flat plate area of the load, given by $\rho_\text{air} = 1.21$ kg/m$^3$ and $a_\text{S} = 0.01$ m$^2$, respectively, and $\bm{\kappa}^\frL = \bm{\kappa}_\text{W}^\frL - \bm{\kappa}_\text{L}^\frL$ the relative wind velocity expressed in $\frL$, with $\bm{\kappa}_\text{W}^\frL = (\RIL)^T \bm{\kappa}_\text{W}^\frI$ the environmental wind velocity expressed in $\frL$, and $\bm{\kappa}_\text{L}^\frL = (\RIL)^T \dxi$ the load velocity expressed in $\frL$. The profile of the environmental wind $\bm{\kappa}_\text{W}^\frI$ is shown in Figure \ref{fig:results_disturbancestrajectory}.}

\begin{table}[!htb]
	\centering
	\caption{Paths composing the reference trajectory.}
	\small
	\begin{tabular}{c c c c}
		Time (sec) & $x^\text{tr}(t)$ (m) & $y^\text{tr}(t)$ (m) & $z^\text{tr}(t)$ (m)\\ \hline
		$0\leq t <10$ & $0.01 t^2 \cos\left(\frac{\pi t}{4}\right)$ & $\sin\left(\frac{\pi t}{20}\right) \sin\left(\frac{\pi t}{4}\right)$ & $3.5 - 2.5 \cos\left(\frac{\pi t}{10}\right)$ \\
		$10\leq t <19$ & $-\frac{\pi}{4}(t-10)$ & $1$ & $6$ \\
		$19\leq t <20$ & $-\frac{9\pi}{4} - 0.5\sin\left(\frac{\pi}{2}(t-19)\right)$ & $1.5 - 0.5\cos\left(\frac{\pi}{2}(t-19)\right)$ & $6$ \\
		$20\leq t <29$ & $-\frac{9\pi}{4} - 0.5$ & $1.5 + \frac{\pi}{4}(t-20)$ & $6$ \\
		$29\leq t <30$ & $-\frac{9\pi}{4} - 0.5\cos\left(\frac{\pi}{2}(t-29)\right)$ & $1.5 {+} \frac{9\pi}{4} {+} 0.5\sin\left(\frac{\pi}{2}(t-29)\right)$ & $6$ \\
		$30\leq t <40$ & $-\frac{9\pi}{4} + \frac{\pi}{4}(t-30)$ & $2 + \frac{9\pi}{4}$ & $6$ \\
		$40\leq t $ & $-\frac{\pi}{80}t^2 + \frac{5\pi}{4}t - \frac{119\pi}{4}$ & $2 + \frac{9\pi}{4}$ & $3.5 {+} 2.5\cos\left(\frac{\pi}{10}(t-40)\right)$ \\		
		\hline
	\end{tabular}
	\label{tab:results_desiredtrajectory}
	\normalsize
\end{table}

\begin{figure}[!htb]
	\begin{footnotesize}
		\centering{
			\def\svgwidth{0.8\textwidth}
			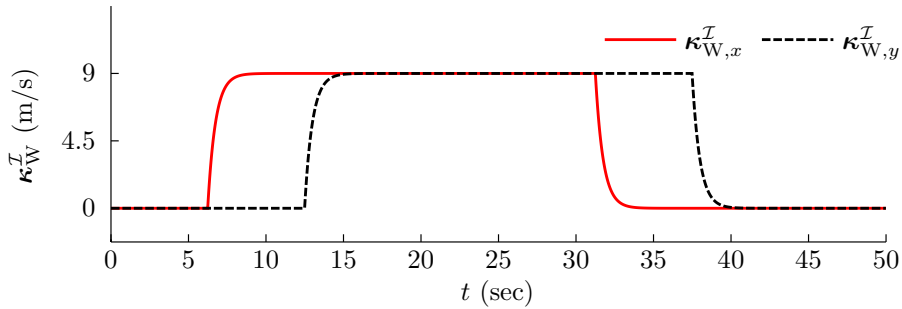
			\caption{\textred{Profile of the enviromental wind disturbances applied to the load.}}\label{fig:results_disturbancestrajectory}}
	\end{footnotesize}
\end{figure}

Table \ref{tab:results_sensorparameters} shows the parameters of the sensors used in the experiment. The noise bounds of the GPS, barometer/IMU, and the servos' sensors were taken from the Novatel OEMStar GPS receiver, Xsens MTi-G, and Herkulex DRS-0101/DRS-0201 sensors datasheets, respectively. The noise bounds of the camera were chosen empirically. The assumptions on probability density functions were made for simulation purposes, being only the knowledge on the noise bounds used for the zonotopic state estimator design. For Gaussian distributions, `noise bound' means three times the standard deviation. The union of all sets $\mathbb{I}$ of sensors whose measurements are available at time instant $k$ yields the set $\mathbb{I}_k$ employed in both state estimation algorithms. \textred{Moreover, although not explicitly taken into account in the experiment, the presence of parameter uncertainties in the system can be expressed implicitly by the considered noise bounds.}

\begin{table}[!htb]
	\footnotesize
	\centering
	\caption{Parameters of the sensors.}
	\begin{tabular}{c c c c c}
		Sensor & $\mathbb{I}$ & Noise bound & Sampling time & PDF (for simulation) \\ \hline
		GPS & $\{1,2\}$ & $\pm\!\!$ $0.15\!$ m & $120\!$ ms & Gaussian \textred{(truncated)}\\ 
		Barometer & $\{3\}$ & $\pm\!\!$ $0.51\!$ m & $12\!$ ms & Gaussian \textred{(truncated)}\\ 
		\multirow{2}{*}{IMU} & $\{4,5,6\}$ & $\pm\!\!$ $2.618 {\cdot} 10^{-3}\!$ rad & \multirow{2}{*}{$12\!$ ms} & \multirow{2}{*}{Gaussian \textred{(truncated)}} \\
		& $\{7,8,9\}$ & $\pm\!\!$ $16.558 {\cdot} 10^{-3}\!$ rad/s & \\ 
		\multirow{2}{*}{Camera} & $\{10,11\}$ & $\pm\!\!$ $0.005\!$ m & \multirow{2}{*}{$24\!$ ms} & \multirow{2}{*}{Uniform} \\
		& $\{12\}$ & $\pm\!\!$ $0.02\!$ m & \\		                           	                           
		\multirow{2}{*}{Servos} & $\{13,14\}$ & $\pm\!\!$ $5.67 {\cdot} 10^{-3}\!$ rad & \multirow{2}{*}{$12\!$ ms} & \multirow{2}{*}{Uniform} \\
		& $\{15,16\}$ & $\pm\!\!$ $0.50772\!$ rad/s & \\		                           
		\hline
	\end{tabular} \normalsize
	\label{tab:results_sensorparameters}
\end{table}

\subsection{\textred{Model, estimator and control design parameters}}

Table \ref{tab:results_modelparameters} shows the model parameters of the tilt-rotor UAV with suspended load. Mass, inertia and displacement parameters were computed from the system's CAD 3D model, designed using the Solidworks\textsuperscript{\textregistered} software. The gravitational acceleration is assumed constant. The experimental parameters $k_\tau$ and $b$ are the same considered in \cite{Almeida2015b}, while $\lambdaR$ and $\lambdaL$ are given according to the direction of rotation of the UAV's propellers: the right one rotates counter-clockwise, and the left one rotates clockwise. The parameters related to viscous friction were chosen empirically.

\begin{table}[!htb]
	\centering
	\caption{\textred{Model parameters of the tilt-rotor UAV with suspended load.}}
	\begin{tabular}{c c}
		Parameter & Value \\ \hline
		$(\mL,m_1,m_2,m_3)$ & $(0.5,1.7068,0.08978,0.08978)$ Kg \\
		$(\dLA{1},\dAB{1})$ & $([ 0 \; 0 \; 0.5 ]^T,[ 0 \; 0 \; 0.119 ]^T)$ m\\
		$\dBC{1}$ & $[ -0.004321 \;\, 0.000601 \;\, -0.045113]^T$ m\\
		$(\dBA{2},\dBA{3})$ & $([  0 \;\, -0.275433 \;\, 0.056262 ]^T, [  0 \;\, 0.275433 \;\, 0.056261 ]^T)$ m\\		
		$(\dAC{2}{2},\dAC{3}{3})$ & $([  0 \; 0 \; 0.056472 ]^T, [  0 \; 0 \; 0.056482 ]^T)$ m\\		
		$\IL$ & $8.333 \cdot 10^{-6} \cdot \eye{3} $ Kg$\cdot$m$^2$ \\
		$\Ii{1}$ & $\begin{bmatrix} 4047.04 & 0.860582 & 9.65766 \\ * &  881.618 & -0.873079 \\ * & * & 4173.18 \end{bmatrix} \cdot 10^{-6}$ Kg${\cdot}$m$^2$ \\
		$\Ii{2}$ & $\begin{bmatrix} 335.737 & -1.33011 {\cdot} 10^{-15} & -2.85046 {\cdot} 10^{-15}\\ * & 335.737 & -6.81597 {\cdot} 10^{-16} \\ * & * & 641.59 \end{bmatrix} \cdot 10^{-6}$ Kg${\cdot}$m$^2$ \\
		$\Ii{3}$ & $\begin{bmatrix} 335.737 & -5.26914 {\cdot} 10^{-16} & 2.9351 {\cdot} 10^{-15} \\ * & 335.737 & -1.24077 {\cdot} 10^{-16} \\ * & * & 641.59 \end{bmatrix} \cdot 10^{-6}$ Kg${\cdot}$m$^2$ \\		
		$\hat{\bm{g}}$ & $[ 0\;\, 0 \;\, -9.81]^T$ m/s$^2$ \\
		$(k_\tau, b)$ &  $(1.7 {\cdot} 10^{-7}$ N$\cdot$m$\cdot$s$^2, 9.5 {\cdot} 10^{-6}$ N$\cdot$s$^2)$\\
		$(\lambdaR, \lambdaL, \beta, \mu_\gamma, \mu_\alpha)$ & $(1, \text{-}1, 5^\text{o}, 0.005\text{ N}{\cdot}\text{m/(rad/s)}, 0.005\text{ N}{\cdot}\text{m/(rad/s)})$\\
		\hline
	\end{tabular}
	\label{tab:results_modelparameters}
\end{table}

Using the presented data, and assuming $\bm{d} = \zeros{3}{1}$, the following equilibrium point was obtained for the nonlinear system \eqref{eq:modeling_statespacenonlinear}:

\begin{equation} \label{eq:results_equilibriumpoint}
	\begin{aligned}
		\bm{q}^{\text{eq}} = &~
		[ 0 \;\, 0 \;\, 0 \;\, 0 \;\, 0 \;\, 0 \;\, 0.00013170 \;\, 0.01396015 \;\, 0.01400528 \;\, 0.01380910]^T \text{,} \\
		\bm{u}^{\text{eq}} = &~
		[ 11.73225673 \;\, 11.76760246 \;\, 4.13886816{\cdot}10^{-7} \;\, 1.01209997{\cdot}10^{-5}]^T \text{.}
	\end{aligned}
\end{equation}


The mixed $\mathcal{H}_2/\mathcal{H}_\infty$ control design was performed using the Yalmip toolbox \citep{Lofberg2004} with the SDPT3 solver \citep{Toh1999}. The design parameters for the LMI regions are given by $\varepsilon = 0.55$, $\varpi = 0.994$ and $\tau = 0.3$. The Bryson's rule \citep{Johnson1987} was used as starting point to synthesize the weighting matrices of the mixed $\mathcal{H}_2/\mathcal{H}_\infty$ controller, which are given by
\begin{align}
	\bm{H}_{\bm{z}} & = \text{diag} \left(\frac{\sqrt{10}}{2},\frac{\sqrt{10}}{2},\frac{\sqrt{10}}{2},\frac{\sqrt{0.5}}{\pi/2},\frac{\sqrt{0.5}}{\pi/2},\frac{\sqrt{5}}{\pi},\frac{1}{\pi/2},\frac{1}{\pi/2},\frac{0.1}{\pi/2},\frac{0.1}{\pi/2},\frac{1}{2},\frac{1}{2},\frac{1}{2},\right.\\
	& \quad \quad \quad \left. \frac{1}{\pi/3},\frac{1}{\pi/3},\frac{1}{\pi/4},\frac{\sqrt{5}}{3\pi},\frac{\sqrt{5}}{3\pi},\frac{0.1}{3\pi},\frac{0.1}{3\pi},\sqrt{5},\sqrt{5},\sqrt{5},\sqrt{0.1}\right) \text{,} \nonumber\\
	\bm{D}_{\bm{zu}} & = \begin{bmatrix}
	\dfrac{\sqrt{750}}{30-\infR^\text{eq}} & 0 & 0 & 0 \\
	0 & \dfrac{\sqrt{750}}{30-\infL^\text{eq}} & 0 & 0 \\
	0 & 0 & \dfrac{\sqrt{5000}}{2-\intauaR^\text{eq}} & 0 \\
	\zeros{2}{1} & \zeros{2}{1} & \zeros{2}{1} & \zeros{2}{1} \\
	0 & 0 & 0 & \dfrac{\sqrt{5000}}{2-\intauaL^\text{eq}} \\
	\zeros{18}{1} & \zeros{18}{1} & \zeros{18}{1} & \zeros{18}{1}
	\end{bmatrix} \text{,} \quad
	\bm{D}_{\bm{zd}} = \begin{bmatrix}
	\zeros{10}{3} \\
	\eye{3} \\
	\bm{N} \\
	0.5{\cdot}\ones{1}{3}\\
	\bm{N} \\
	\zeros{2}{3} \\
	\eye{3} \\
	0.5{\cdot}\ones{1}{3}
\end{bmatrix} \text{,}
\end{align}
with $\bm{N} \triangleq \begin{bmatrix} 0 & 1 & 0 \\ 1 & 0 & 0 \end{bmatrix}$, $\infR^\text{eq}$, $\infL^\text{eq}$, $\intauaR^\text{eq}$ and $\intauaL^\text{eq}$ are equilibrium values from \eqref{eq:results_equilibriumpoint}, and $\bm{D}_{\bm{zd}}$ was adjusted by trial and error. The chosen upper bound for the $\mathcal{H}_\infty$ norm is given by $\|\bm{\Psi}^{(\infty)}_{\bm{d}\bm{z}}\|_\infty^2 < 81$, for which $\|\bm{\Psi}^{(2)}_{\bm{d}\bm{z}}\|_2^2 < 2.7643$. Moreover, due to infeasibility issues in solving the optimization problem for the uncertain system, when considering the whole range of desired accelerations for the load, these were assumed to be bounded by $\ddot{x}^\text{tr}(t) \in [-0.5,\; 0.5]$, $\ddot{y}^\text{tr}(t) \in [-0.5,\; 0.5]$, and $\ddot{z}^\text{tr}(t) \in [-0.3,\; 0.3]$.


For the zonotopic state estimator, the chosen initial zonotope $\bar{\mathbb{X}}_0$ is a box centered at the desired trajectory, given by $\bar{\mathbb{X}}_0 = [(\bm{\xi}^\text{tr}_0)^T \; \zeros{1}{20}]^T \oplus \bm{G}_{\bar{x}_0} \ubox^{23}$, with generator matrix \textred{$\bm{G}_{\bar{x}_0} = \text{diag} \left( 0.5 {\cdot} \ones{3}{1}, \frac{\pi}{4} {\cdot} \ones{7}{1}, \ones{13}{1} \right)$}.
%
To prevent its complexity from increasing indefinitely, the order of the estimated zonotope $\hat{\mathbb{X}}_k$ was limited to 75 times its dimension. Moreover, the zonotopes $\bar{\mathbb{W}}$ and $\bar{\mathbb{V}}$ were adjusted as $\bar{\mathbb{W}} =  \zeros{23}{1} \oplus \bm{G}_{\bar{w}} \ubox^{23}$ and $\bar{\mathbb{V}} = \bm{\pi}(\bm{x}^\text{eq}) \oplus \bm{G}_{\bar{v}} \ubox^{16}$, with generator matrices\footnote{The zonotope $\bar{\mathbb{V}}$ was chosen using the noise bounds from Table \ref{tab:results_sensorparameters} as starting point, then adjusted empirically in order to accommodate the unmodeled dynamics due to linearization.}
\begin{align}
	\textred{\bm{G}_{\bar{w}}} = & ~ \text{diag} (10^{-4} {\cdot} \ones{8}{1}, 1.5 {\cdot} 10^{-4}{\cdot} \ones{2}{1}, 10^{-4} {\cdot} \ones{3}{1}, 0.01 {\cdot} \ones{3}{1}, 0.05 {\cdot} \ones{2}{1}, 10^{-4} {\cdot} \ones{2}{1}, \\ & \quad \quad ~~ 0.01 {\cdot} \ones{3}{1}) \text{,} \nonumber\\
	\bm{G}_{\bar{v}} = & ~ \text{diag} (0.18 {\cdot} \ones{2}{1}, 0.612, 3.1416\ten{-3},3.1416\ten{-3}, 0.03, 19.872 {\cdot} 10^{-3}, 19.872 {\cdot} 10^{-3}, \\ & \quad \quad ~~ 0.24, 0.006 {\cdot} \ones{2}{1}, 0.06 ,  6.8067 {\cdot} 10^{-3} {\cdot} \ones{2}{1}, 0.6093 {\cdot} \ones{2}{1}) \text{.} \nonumber
\end{align}

\textred{For comparison purposes, the initial states of the Kalman filter are given by $\Delta \hat{\bm{x}}_{0|0} = [(\bm{\xi}^\text{tr}_0)^T \; \zeros{1}{20}]^T$, and the covariance matrices are diagonal, computed based on the generator matrices of the zonotopic state estimator, by regarding radius as three times the corresponding standard deviation, as $\bm{P}^{\bm{\nu}}_{0|0}(i,i) = (\bm{G}_{\bar{x}_0}(i,i)/3)^2$, $\bm{P}^{\bm{w}}(i,i) = (\bm{G}_{\bar{w}}(i,i)/3)^2$ and $\bm{P}^{\bm{v}}(i,i) = (\bm{G}_{\bar{v}}(i,i)/3)^2$, from which $\bm{P}^{\bm{v}[k]}$ is formed using all $i \in \setI_k$ (see Section \ref{sec:estimation_LKF})}.


\subsection{Experiment results and discussion}

The trajectories performed by the UAV and the load are shown in Figure \ref{fig:results_performedTrajectories}.\footnote{See accompanying video.} The path tracking of the load was performed with success, from take-off to landing, using the proposed zonotopic state estimator and the mixed $\mathcal{H}_2/\mathcal{H}_\infty$ controller. Moreover, despite temporary deviations from the desired trajectory, the disturbances affecting the load were rejected in steady-state, as shown by the evolution of the tracking error, depicted in Figure \ref{fig:results_trackingError}. Even under different sampling times and non-Gaussian measurement noise (see Table \ref{tab:results_sensorparameters}), the proposed zonotopic state estimator was capable of providing the system states to the mixed $\mathcal{H}_2/\mathcal{H}_\infty$ controller. 


\begin{figure}[!htb]
	\begin{scriptsize}
			\centering{
			\def\svgwidth{0.8\textwidth}
			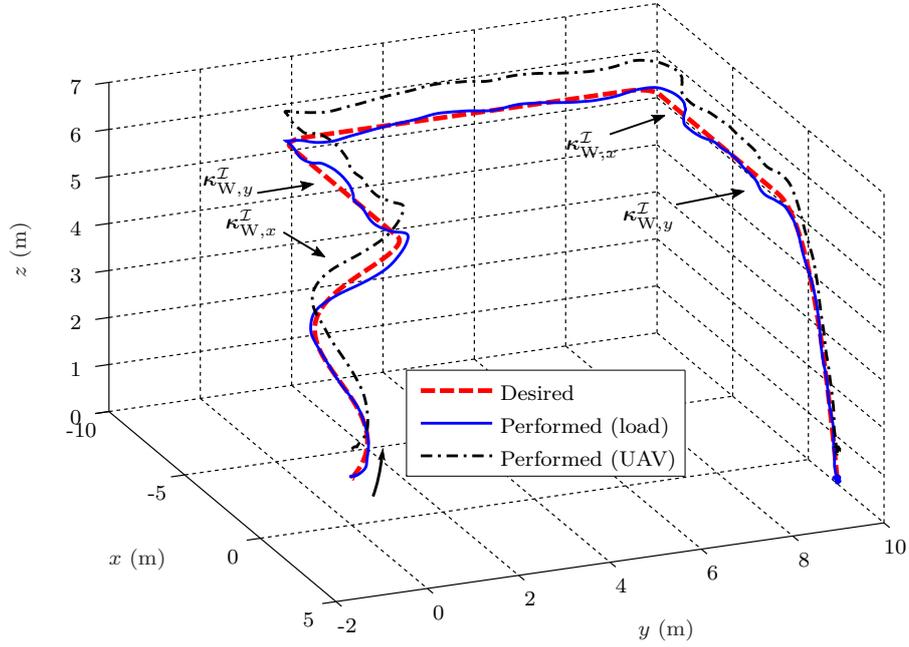
		\caption{\textred{Trajectories performed by the UAV and the load using the ZSE in the ProVANT simulator.}}\label{fig:results_performedTrajectories}}
	\end{scriptsize}
\end{figure}

\begin{figure}[!htb]
	\begin{scriptsize}
		\centering{
			\def\svgwidth{0.8\textwidth}
			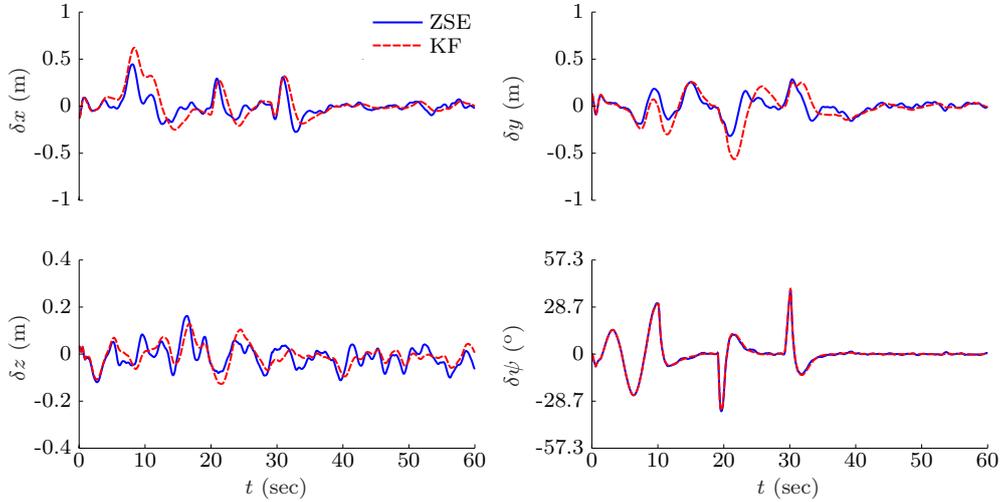
			\caption{Tracking error of the regulated variables.}\label{fig:results_trackingError}}
	\end{scriptsize}
\end{figure}

Figure \ref{fig:results_remainingDOF} shows the time evolution of the remaining degrees of freedom of the system, which were kept stable as the trajectory was performed by the load. Through these results, one can conclude that the UAV remained stable as well, since the aircraft's behavior with respect to the inertial frame is described implicitly by these variables. Moreover, the designed mixed $\mathcal{H}_2/\mathcal{H}_\infty$ controller was able to stabilize the aircraft without the need of a cascade control structure. Figure \ref{fig:results_inputs} shows the actuator signals generated by the mixed $\mathcal{H}_2/\mathcal{H}_\infty$ control law. Despite the existing noise, the inertial properties of the aircraft actuators would straightforwardly attenuate such noise in a physical setup.

\begin{figure}[!htb]
	\begin{scriptsize}
		\centering{
			\def\svgwidth{0.85\textwidth}
			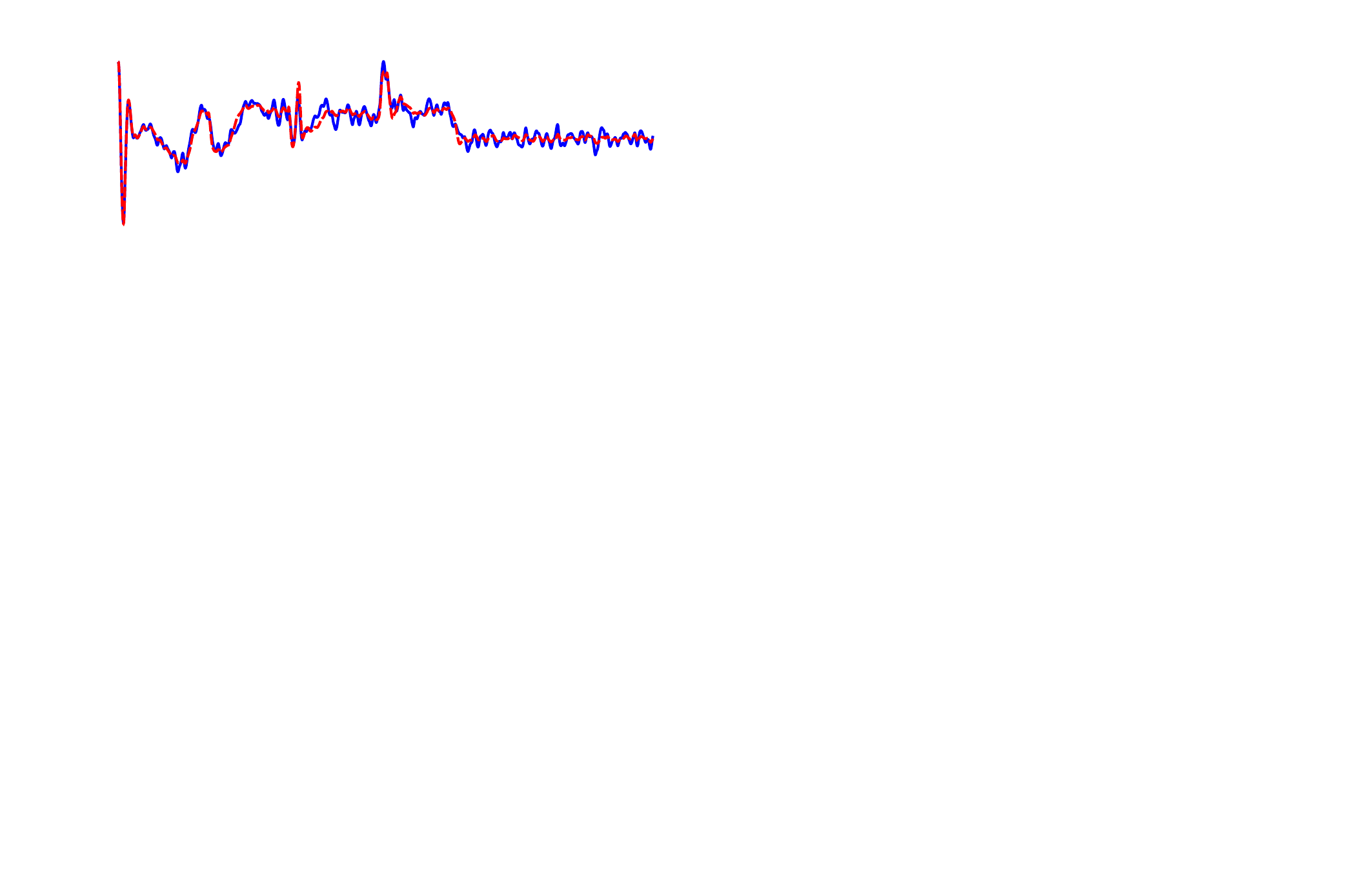
			\caption{Time evolution of the remaining degrees of freedom.}\label{fig:results_remainingDOF}}
	\end{scriptsize}
\end{figure}

Figure \ref{fig:results_ZSE_estimationGenCoord} shows the estimation error of the generalized coordinates using the proposed zonotopic state estimator. As expected from Algorithm \ref{alg:estimation_ZSEtiltrotor}, the estimation error remained inside the associated confidence limits\footnote{These limits were obtained through the interval hull of the estimated zonotope \cite{Kuhn1998}.}, which implies that the real states remained inside the estimated zonotope during the experiment. Thus, the proposed zonotopic state estimator was able to provide the system states with consistency. Moreover, some patterns can be noted in the confidence limits, which were generated by the different sampling times of the available sensors. These also appear in the time evolution of the Frobenius norm of the generator matrix, depicted in Figure \ref{fig:results_estimationFrobNorm}. 

\begin{figure}[!htb]
	\begin{scriptsize}
		\centering{
			\def\svgwidth{0.8\textwidth}
			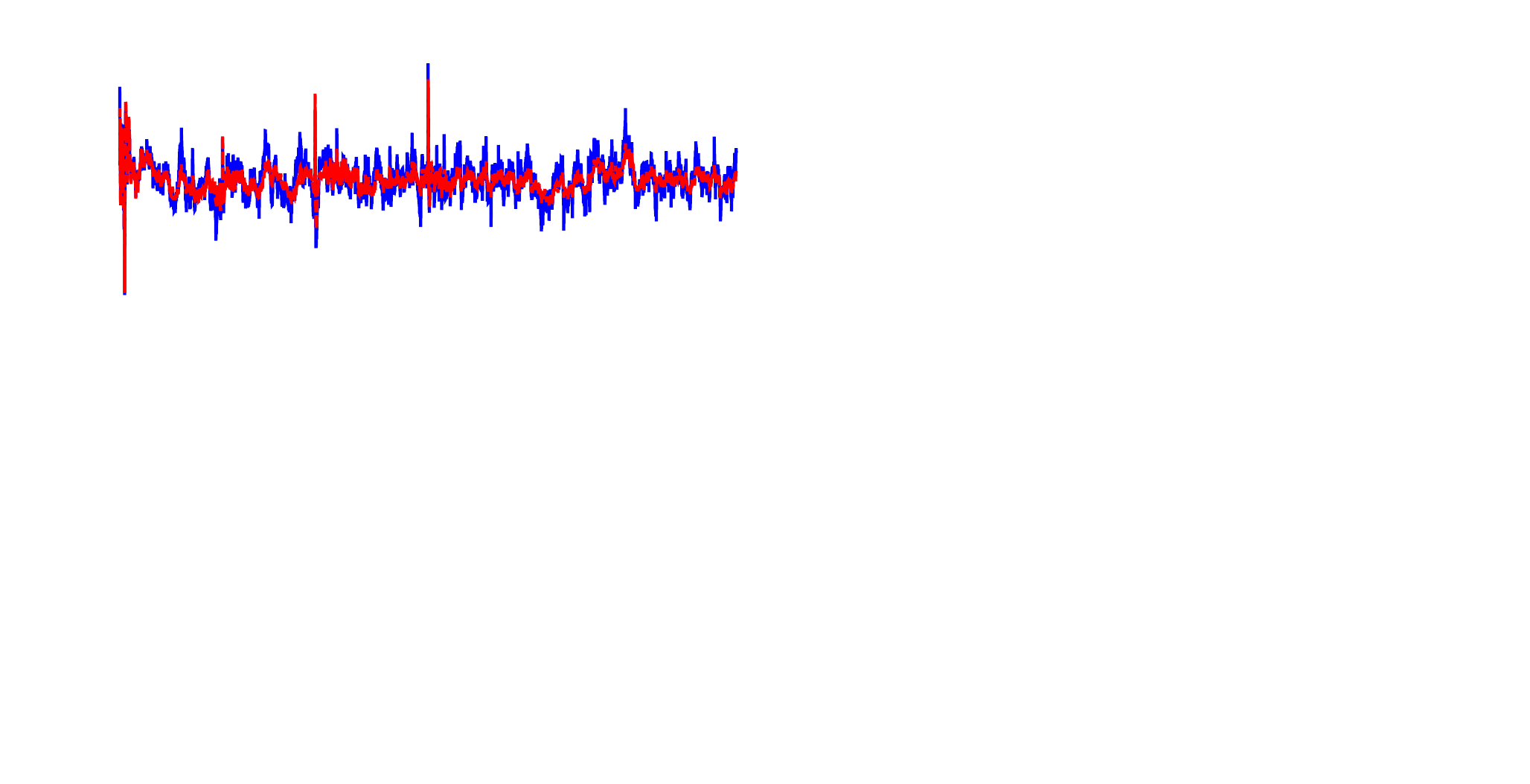
			\caption{Control signals applied to the aircraft's actuators, generated by the mixed $\mathcal{H}_2/\mathcal{H}_\infty$ controller.}\label{fig:results_inputs}}
	\end{scriptsize}
\end{figure}

\begin{figure}[!htb]
	\begin{footnotesize}
		\centering{
			\def\svgwidth{0.75\textwidth}
			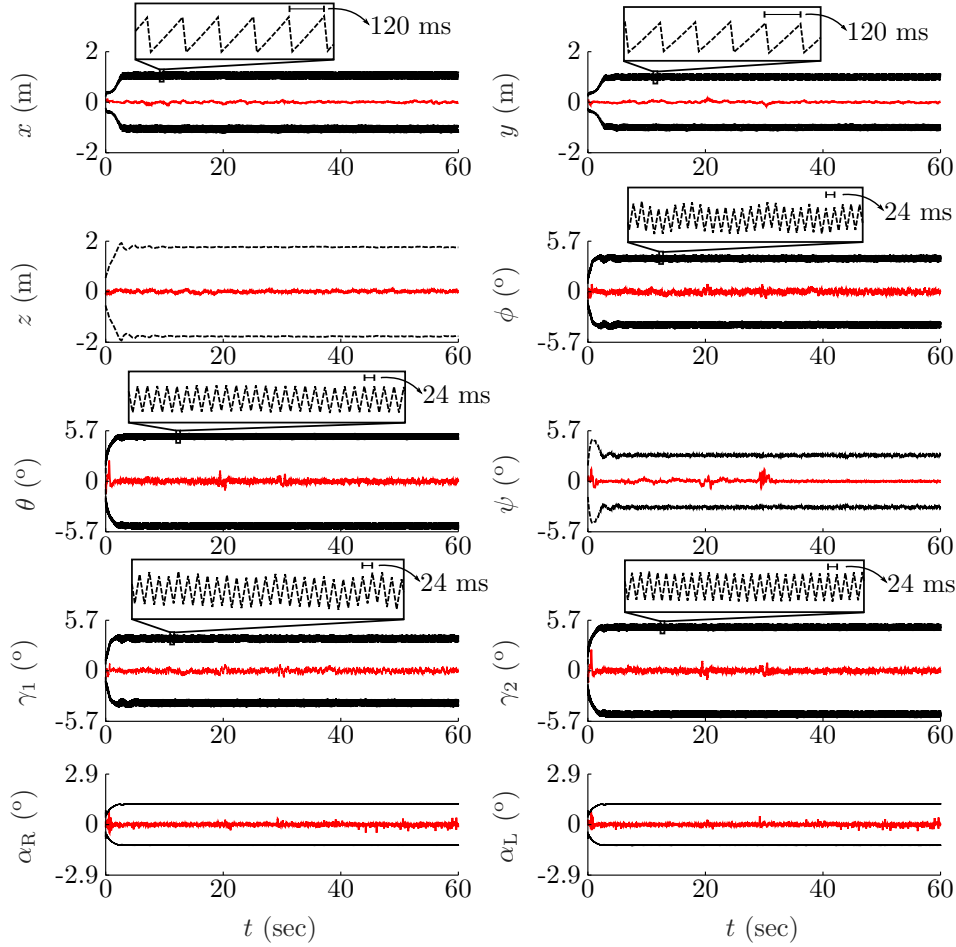
			\caption{Estimation error of the generalized coordinates using the zonotopic state estimator. Solid lines denote estimation error, while dashed lines denote confidence limits.}\label{fig:results_ZSE_estimationGenCoord}}
	\end{footnotesize}
\end{figure}

\begin{figure}[!htb]
	\begin{footnotesize}
		\centering{
			\def\svgwidth{0.85\textwidth}
			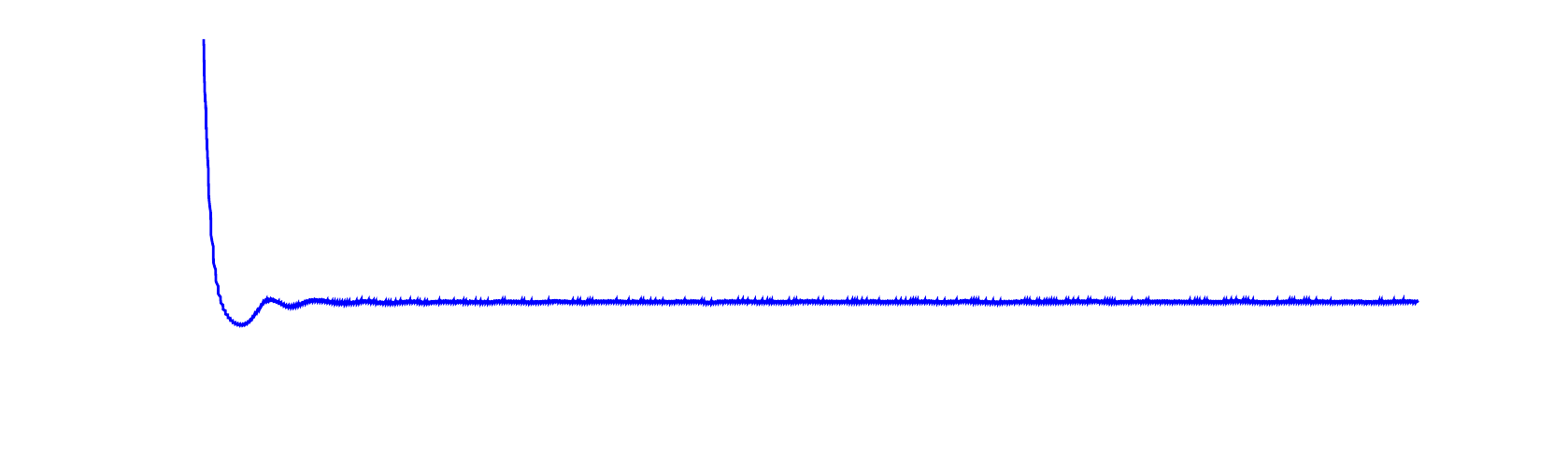
			\caption{Frobenius norm of the estimated zonotope's generator matrix.}\label{fig:results_estimationFrobNorm}}
	\end{footnotesize}
\end{figure}

Regarding the results obtained using the Kalman filter, the associated tracking error is shown in Figure \ref{fig:results_trackingError}. Furthermore, the time evolution of the remaining degrees of freedom, as well as the associated control signals, are also shown in Figures \ref{fig:results_remainingDOF} and \ref{fig:results_inputs}, respectively. The estimation error is shown in Figure \ref{fig:results_LKF_estimationGenCoord}, in which the patterns due to different sampling times also appear. It is apparent that the time evolution of some variables were less noisy in comparison with the zonotopic state estimator experiment. This result is expected, since the Kalman filter provides minimum-variance estimation. 

\textred{For comparison purposes, the Root Mean Square Error (RMSE) was computed for the tracking error and estimation error of the regulated variables for both estimators, and is shown in Table \ref{tab:RMSE}. In terms of the RMSE, the zonotopic state estimator provided considerably better estimates for $x$, $y$, while the Kalman filter provided slightly better estimates for the altitude $z$ and yaw angle $\psi$. Nevertheless, despite the tighter confidence limits (see Figure \ref{fig:results_LKF_estimationGenCoord}), the Kalman filter was unable to estimate all the system states with consistency, depicted by the fact that the estimation error for most variables exceeded the confidence limits of the filter in several points during the experiment. Finally, note that the most critical deviations of the Kalman filter occurred during the initial transient due to the non-equilibrium initial conditions, and also during rapid changes in the trajectory tracking direction, which demonstrate its inability to cope with the nonlinearities of the system effectively.}

\begin{table}[!htb]
	\centering
	\caption{\textred{Root mean square error of the regulated variables.}}
	\begin{tabular}{c | c c c | c c c}
		\hline
		\multirow{2}{*}{Variable} & \multicolumn{3}{c}{Estimation error} & \multicolumn{3}{c}{Tracking error} \\ \cline{2-7}
		& ZSE & KF & ZSE/KF & ZSE & KF & ZSE/KF\\
		\hline
		$x$ & 0.0300 & 0.0679 & 44.15\% & 0.1064 & 0.1546 & 68.82\% \\
		$y$ & 0.0315 & 0.0697 & 45.13\% & 0.1011 & 0.1504 & 67.23\% \\
		$z$ & 0.0342 & 0.0289 & 118.27\% & 0.0502 & 0.0466 & 107.71\% \\
		$\psi$ & 0.0030 & 0.0027 & 112.07\% & 0.1547 & 0.1548 & 99.95\% \\
		\hline
	\end{tabular}
	\label{tab:RMSE}
\end{table}



\begin{figure}[!htb]
	\begin{footnotesize}
		\centering{
			\def\svgwidth{0.8\textwidth}
			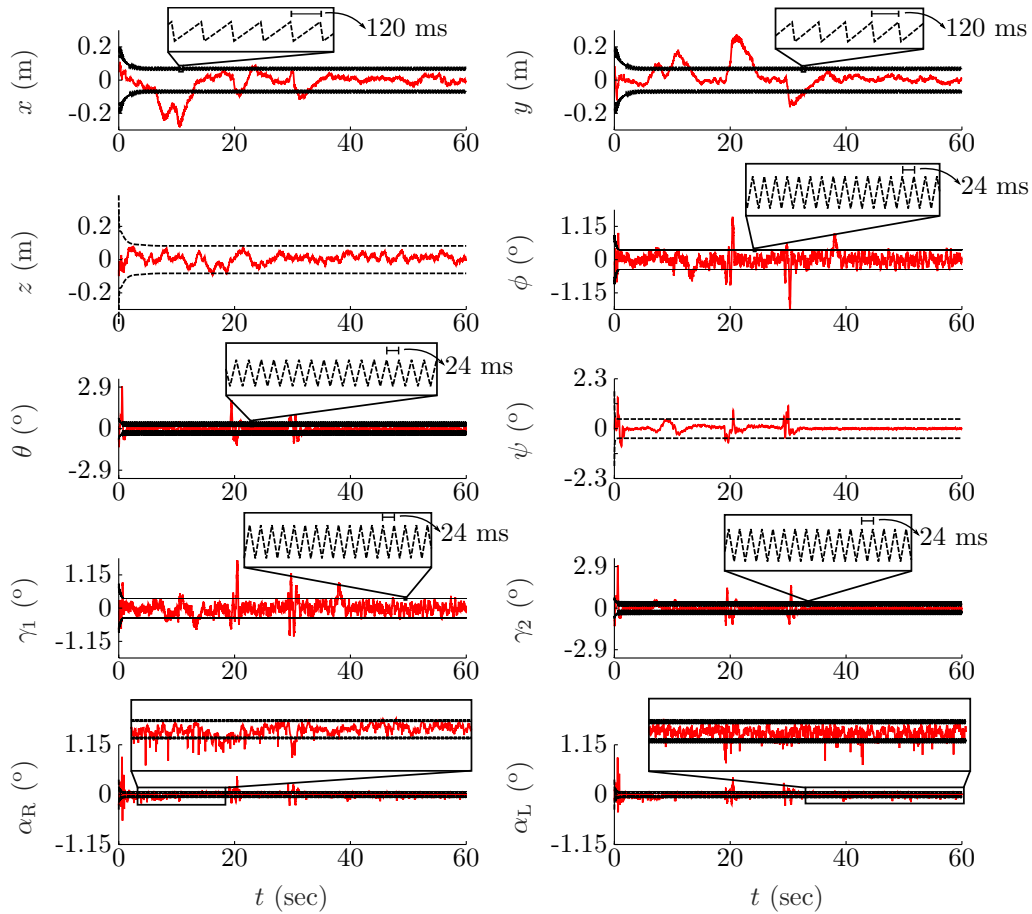
			\caption{Estimation error of the generalized coordinates using the Kalman filter. Solid lines denote estimation error, while dashed lines denote confidence limits.}\label{fig:results_LKF_estimationGenCoord}}
	\end{footnotesize}
\end{figure}


\section{Conclusions and future work} \label{sec:conclusions}

This paper dealt with the problem of trajectory tracking of a suspended load using a tilt-rotor UAV as mobile platform. In order to solve the challenge, a detailed model of the system was developed from the perspective of the load, with a reduced number of assumptions in comparison to previous works. The position and orientation of the load were chosen as degrees of freedom of the system, yielding a nonlinear state-space representation with these variables among the system states, with the UAV's position and orientation being described with respect to the load. A discrete-time state-feedback mixed $\mathcal{H}_2/\mathcal{H}_\infty$ control strategy was designed for path tracking of the suspended load, with disturbance compensation and guaranteed transient response properties, with an enlarged domain of attraction, by considering the desired accelerations for the load in the control design through an uncertain linear parameter-varying framework.

Considering a recurrent scenario in which available information is provided by sensors embedded at the UAV, a set-membership zonotopic state estimator was designed, to provide the load's position and orientation, which was able to cope with different sampling times and unknown-but-bounded uncertainties. To demonstrate and compare the performance of the proposed state estimator, a Kalman filter was also designed. The performance of the proposed strategies were demonstrated through numerical experiments, performed in the ProVANT simulation environment, a platform based on the Gazebo simulator and on a CAD 3D model of the system. In contrast to the Kalman filter, the zonotopic state estimator was able to provide the system states with consistency. The path tracking was performed with success by the suspended load, from take-off to landing, using the designed mixed $\mathcal{H}_2/\mathcal{H}_\infty$ control strategy with the proposed control structure.

This work dealt with hover flights, with the tilt-rotor UAV in helicopter flight-mode. For future works, the inclusion of aerodynamic
surfaces into the proposed model is straightforward, and an enlarged flight envelope will be explored. Another future research step consists in validating the designed strategies in an experimental setup. 

\bibliography{thesis_bibliography}

\begin{thebibliography}{10}
\expandafter\ifx\csname url\endcsname\relax
  \def\url#1{\texttt{#1}}\fi
\expandafter\ifx\csname urlprefix\endcsname\relax\def\urlprefix{URL }\fi
\expandafter\ifx\csname href\endcsname\relax
  \def\href#1#2{#2} \def\path#1{#1}\fi

\bibitem{Ngo2012}
Q.~H. Ngo, K.-S. Hong, Sliding-mode antisway control of an offshore container
  crane, IEEE/AMSE Transactions on Mechatronics 17~(2) (2012) 201--209.

\bibitem{Bernard2011}
M.~Bernard, K.~Kondak, I.~Maza, A.~Ollero, Autonomous transportation and
  deployment with aerial robots for search and rescue missions, Journal of
  Field Robotics 28~(6) (2011) 914--931.

\bibitem{BisgaardThesis}
M.~Bisgaard, Modeling, estimation and control of helicopter slung load system,
  Ph.D. thesis, Aalborg University (2008).

\bibitem{Wu2015}
Z.~Wu, X.~Xia, B.~Zhu, Model prediction control for improving operational
  efficiency of overhead cranes, Nonlinear Dynamics 79~(4) (2015) 2639--2657.

\bibitem{Chen2007}
S.~J. Chen, B.~Hein, H.~Wörn, Swing attenuation of suspended objects
  transported by robot manipulator using acceleration compensation, in: 2007
  IEEE/RSJ International Conference on Intelligent Robots and Systems, 2007,
  pp. 2919--2924.

\bibitem{Bisgaard2009b}
M.~Bisgaard, A.~la~Cour-Harbo, J.~D. Bendtsen, Swing damping for helicopter
  slung load systems using delayed feedback, in: AIAA Guidance, Navigation, and
  Control Conference, 2009, pp. 1--11.

\bibitem{Fusato2001}
D.~Fusato, G.~Guglieri, R.~Celi, Flight dynamics of an articulated rotor
  helicopter with an external slung load, Journal of the American Helicopter
  Society 46~(1) (2001) 3--13.

\bibitem{JainThesis}
R.~P.~K. Jain, Transportation of a cable suspended load using unmanned aerial
  vehicles, Master's thesis, Delft University of Technology (2015).

\bibitem{Bisgaard2007a}
M.~Bisgaard, A.~la~Cour-Harbo, E.~N. Johnson, J.~D. Bendtsen, Vision aided
  state estimator for helicopter slung load system, in: 17th IFAC Symposium on
  Automatic Control in Aerospace, 2007, pp. 425--430.

\bibitem{Bisgaard2007b}
M.~Bisgaard, A.~la~Cour-Harbo, E.~N. Johnson, J.~D. Bendtsen, Full state
  estimation for helicopter slung load system, in: AIAA Guidance, Navigation
  and Control Conference and Exhibit, 2007, pp. 1--15.

\bibitem{Alamo2005a}
T.~Alamo, J.~Bravo, E.~Camacho, Guaranteed state estimation by zonotopes,
  Automatica 41~(6) (2005) 1035--1043.

\bibitem{Le2013}
V.~T.~H. Le, C.~Stoica, T.~Alamo, E.~F. Camacho, D.~Dumur, Zonotopic guaranteed
  state estimation for uncertain systems, Automatica 49~(11) (2013) 3418--3424.

\bibitem{Palunko2012b}
I.~Palunko, P.~Cruz, R.~Fierro, Agile load transportation: safe and efficient
  load manipulation with aerial robots, IEEE Robotics \& Automation Magazine
  19~(3) (2012) 69--79.

\bibitem{Almeida2015b}
M.~M. Almeida, G.~V. Raffo, Nonlinear control of a tiltrotor {UAV} for load
  transportation, in: 11th IFAC Symposium on Robot Control, 2015, pp. 234--239.

\bibitem{Raffo2016}
G.~V. Raffo, M.~M. Almeida, Nonlinear robust control of a quadrotor {UAV} for
  load transportation with swing improvement, in: American Control Conference,
  2016, pp. 3156--3162.

\bibitem{Santos2016b}
M.~A. Santos, G.~V. Raffo, Path tracking model predictive control of a
  tilt-rotor uav carrying a suspended load, in: Proc. of the IEEE 19th
  International Conference on Intelligent Transportation Systems (ITSC), 2016,
  pp. 1458--1463.

\bibitem{Oktay2013}
T.~Oktay, C.~Sultan, Modeling and control of a helicopter slung-load system,
  Aerospace Science and Technology 29~(1) (2013) 206--222.

\bibitem{Liang2018}
X.~Liang, Y.~Fang, N.~Sun, Nonlinear hierarchical control for unmanned
  quadrotor transportation systems, IEEE Transactions on Industrial Electronics
  65~(4) (2018) 3395--3405.

\bibitem{Sanchez2017}
M.~E. Guerrero-Sanchez, D.~A. Mercado-Ravell, R.~Lozano, C.~D.
  García-Beltrán, Swing-attenuation for a quadrotor transporting a
  cable-suspended payload, IEEE Transactions on Industrial Electronics 68~(1)
  (2017) 433--449.

\bibitem{CourHarbo2009}
A.~la~Cour-Harbo, M.~Bisgaard, State-control trajectory generation for
  helicopter slung load system using optimal control, in: AIAA Guidance,
  Navigation, and Control Conference, 2009, pp. 1--16.

\bibitem{Tang2015}
S.~Tang, V.~Kumar, Mixed integer quadratic program trajectory generation for a
  quadrotor with a cable-suspended payload, in: IEEE International Conference
  on Robotics and Automation, 2015, pp. 2216--2222.

\bibitem{Bernard2009}
M.~Bernard, K.~Kondak, Generic slung load transportation system using small
  size helicopters, in: IEEE International Conference on Robotics and
  Automation, 2009, pp. 3258--3264.

\bibitem{Lee2013}
T.~Lee, K.~Sreenath, V.~Kumar, Geometric control of cooperating multiple
  quadrotor {UAVs} with a suspended payload, in: 52nd IEEE Conference on
  Decision and Control, 2013, pp. 5510--5515.

\bibitem{Palunko2013}
I.~Palunko, A.~Faust, P.~Cruz, L.~Tapia, R.~Fierro, A reinforcement learning
  approach towards autonomous suspended load manipulation using aerial robots,
  in: IEEE International Conference on Robotics and Automation, 2013, pp.
  4881--4886.

\bibitem{Sreenath2013b}
K.~Sreenath, T.~Lee, V.~Kumar, Geometric control and differential flatness of a
  quadrotor {UAV} with a cable-suspended load, in: 52nd IEEE Conference on
  Decision and Control, 2013, pp. 2269--2274.

\bibitem{Pereira2016b}
P.~O. Pereira, M.~Herzog, D.~V. Dimarogonas, Slung load transportation with a
  single aerial vehicle and disturbance rejection, in: 24th Mediterranean
  Conference on Control and Automation, 2016, pp. 671--676.

\bibitem{Amiri2011}
N.~Amiri, A.~Ramirez-Serrano, R.~Davies, Modelling of opposed lateral and
  longitudinal tilting dual-fan unmanned aerial vehicle, in: Proc. of the 18th
  IFAC World Congress, 2011, pp. 2054--2059.

\bibitem{Park2013}
S.~Park, J.~Bae, Y.~Kim, S.~Kim, Fault tolerant flight control system for the
  tilt-rotor uav, Journal of the Franklin Institute 350~(9) (2013) 2535--2559.

\bibitem{Cardoso2016}
D.~N. Cardoso, G.~V. Raffo, S.~Esteban, A robust adaptive mixing control for
  improved forward flight of a tilt-rotor {UAV}, in: Proc. of 19th
  International Conference on Intelligent Transportation Systems, 2016, pp.
  1432--1437.

\bibitem{Santos2017}
M.~A. Santos, B.~S. Rego, G.~V. Raffo, A.~Ferramosca, Suspended load path
  tracking control strategy using a tilt-rotor {UAV}, Journal of Advanced
  Transportation 2017 (2017) 1--22.

\bibitem{Rego2016c}
B.~S. Rego, G.~V. Raffo, Suspended load path tracking control based on
  zonotopic state estimation using a tilt-rotor {UAV}, in: Proc. of the IEEE
  19th International Conference on Intelligent Transportation Systems, 2016,
  pp. 1445--1451.

\bibitem{Raffo2011}
G.~V. Raffo, M.~G. Ortega, F.~R. Rubio, Nonlinear $\mathcal{H}_\infty$
  controller for the quad-rotor helicopter with input coupling, in: Proc. of
  the 18th World Congress of the IFAC, 2011, pp. 13834--13839.

\bibitem{Siciliano2009}
B.~Siciliano, L.~Sciavicco, L.~Villani, G.~Oriolo, Robotics: modelling,
  planning and control, Springer Science \& Business Media, 2009.

\bibitem{Shabana2010}
A.~A. Shabana, Computational Dynamics, 3rd Edition, John Wiley and Sons, 2010.

\bibitem{Kane1985}
T.~R. Kane, D.~A. Levinson, Dynamics: Theory and Applications, Mcgraw-Hill
  College, 1985.

\bibitem{Castillo2005}
P.~Castillo, R.~Lozano, A.~Dzul, Stabilization of a mini rotorcraft with four
  rotors, IEEE control systems magazine 25~(6) (2005) 45--55.

\bibitem{Moore2009}
R.~E. Moore, R.~B. Kearfott, M.~J. Cloud, Introduction to Interval Analysis,
  SIAM, Philadelphia, PA, USA, 2009.

\bibitem{Kuhn1998}
W.~K{\"u}hn, Rigorously computed orbits of dynamical systems without the
  wrapping effect, Computing 61~(1) (1998) 47--67.

\bibitem{Combastel2003}
C.~Combastel, A state bounding observer based on zonotopes, in: 2003 European
  Control Conference, 2003, pp. 2589--2594.

\bibitem{Simon2006}
D.~Simon, Optimal State Estimation: {Kalman}, $H_\infty$ and Nonlinear
  Approaches, John Wiley \& Sons, Inc, 2006.

\bibitem{Raffo2018}
G.~Raffo, M.~Almeida, A load transportation nonlinear control strategy using a
  tilt-rotor uav, Journal of Advanced Transportation 2018 (2018) 1--20.

\bibitem{Gahinet1996}
M.~Chilali, P.~Gahinet, $\mathcal{H}_\infty$ design with pole placement
  constraints: an lmi approach, IEEE Transactions on Automatic Control 41~(3)
  (1996) 358--367.

\bibitem{Oliveira2002}
M.~C. de~Oliveira, J.~C. Geromel, J.~Bernussou, Extended $\mathcal{H}_2$ and
  $\mathcal{H}_\infty$ norm characterizations and controller parametrizations
  for discrete-time systems, International Journal of Control 75~(9) (2002)
  666--679.

\bibitem{Quigley2009}
M.~Quigley, B.~Gerkey, K.~Conley, J.~Faust, T.~Foote, J.~Leibs, E.~Berger,
  R.~Wheeler, A.~Ng, {ROS}: an open-source robot operating system, in: Proc. of
  the Open-Source Software Workshop of the International Conference on Robotics
  and Automation, 2009, pp. 2054--2059.

\bibitem{Koenig2004}
N.~Koenig, A.~Howard, Design and use paradigms for {Gazebo}, an open-source
  multi-robot simulator, in: Proc. of the 2004 IEEE/RSJ International
  Conference on Intelligent Robots and Systems, 2004, pp. 2149--2154.

\bibitem{Lofberg2004}
J.~L{\"o}fberg, Yalmip: A toolbox for modeling and optimization in matlab, in:
  2004 IEEE International Symposium on Computer Aided Control Systems Design,
  2004, pp. 284--289.

\bibitem{Toh1999}
K.-C. Toh, M.~J. Todd, R.~H. T{\"u}t{\"u}nc{\"u}, Sdpt3 -- a matlab software
  package for semidefinite programming, version 1.3, Optimization methods and
  software 11~(1-4) (1999) 545--581.

\bibitem{Johnson1987}
M.~A. Johnson, M.~J. Grimble, Recent trends in linear optimal quadratic
  multivariable control system design, IEE Proceedings D - Control Theory and
  Applications 134~(1) (1987) 53--71.

\end{thebibliography}

\end{document}